\documentclass[10pt,journal,compsoc]{IEEEtran}
\usepackage{cite}
\usepackage{amsmath,amssymb,amsfonts}
\usepackage{algorithmicx}
\usepackage[ruled, linesnumbered, vlined]{algorithm2e}
\usepackage{array}
\usepackage{multicol}
\usepackage{graphicx}
\usepackage{stfloats}
\usepackage{hyperref}
\usepackage{doi}
\usepackage{amsthm}
\usepackage{subfigure}
\usepackage{url}
\usepackage{verbatim}
\usepackage{algpseudocode}  
\usepackage{textcomp}
\usepackage{xcolor}
\usepackage{color}
\newenvironment{sequation}{\begin{equation}\small}{\end{equation}}

\newtheorem{theorem}{\textbf{Theorem}}

\newtheorem{definition}{\textbf{Definition}}
\definecolor{r}{rgb}{1, 0, 0}
\definecolor{b}{rgb}{0, 0, 1}
\def\BibTeX{{\rm B\kern-.05em{\sc i\kern-.025em b}\kern-.08em
		T\kern-.1667em\lower.7ex\hbox{E}\kern-.125emX}}

\begin{document}
\pagestyle{empty} 

\title{\fontsize{20pt}{24pt}\selectfont Digital Twin-Assisted Space-Air-Ground Integrated Multi-Access Edge Computing for Low-Altitude Economy: An Online Decentralized Optimization Approach}

\author{Long~He,
        Geng~Sun,~\IEEEmembership{Senior Member,~IEEE},
        Zemin Sun, 
        Jiacheng~Wang, 
        Hongyang~Du, \\
        Dusit Niyato,~\IEEEmembership{Fellow,~IEEE},
        Jiangchuan~Liu,~\IEEEmembership{Fellow,~IEEE}, and 
        Victor C. M. Leung,~\IEEEmembership{Life Fellow,~IEEE}
	\thanks{This study is supported in part by the National Natural Science Foundation of China (62272194, 62471200), and in part by the Science and Technology Development Plan Project of Jilin Province (20230201087GX). (\textit{Corresponding author: Geng Sun.)}}
 \IEEEcompsocitemizethanks{
    \IEEEcompsocthanksitem Long He, and Zemin Sun are with the College of Computer Science and Technology, Jilin University, Changchun 130012, China (e-mail: helong23@mails.jlu.edu.cn, sunzemin@jlu.edu.cn).
     \IEEEcompsocthanksitem Geng Sun is with the College of Computer Science and Technology, Jilin University, Changchun 130012, China, and also with the College of Computing and Data Science, Nanyang Technological University, Singapore 639798 (e-mail: sungeng@jlu.edu.cn).
     \IEEEcompsocthanksitem Jiacheng Wang and Dusit Niyato are with the College of Computing and Data Science, Nanyang Technological University, Singapore 639798 (e-mail: jiacheng.wang@ntu.edu.sg, dniyato@ntu.edu.sg).
    \IEEEcompsocthanksitem Hongyang Du is with the Department of Electrical and Electronic Engineering, University of Hong Kong, Pok Fu Lam, Hong Kong (e-mail: duhy@eee.hku.hk).
    \IEEEcompsocthanksitem Jiangchuan Liu is with the School of Computing Science, Simon Fraser University, Burnaby, BC V5A 1S6, Canada, and also with the R\&D Department, Jiangxing Intelligence Inc., Nanjing 210000, China (e-mail: jcliu@sfu.ca).
    \IEEEcompsocthanksitem Victor C. M. Leung is with the Artificial Intelligence Research Institute, Shenzhen MSU-BIT University, Shenzhen 518115, China, with the College of Computer Science and Software Engineering, Shenzhen University, Shenzhen 518060, China, and also with the Department of Electrical and Computer Engineering, The University of British Columbia, Vancouver V6T 1Z4, Canada (e-mail: vleung@ieee.org).}}
    
\IEEEtitleabstractindextext{
%
%
\begin{abstract}
\par The emergence of space-air-ground integrated multi-access edge computing (SAGIMEC) networks opens a significant opportunity for the rapidly growing low altitude economy (LAE), facilitating the development of various applications by offering efficient communication and computing services. However, the heterogeneous nature of SAGIMEC networks, coupled with the stringent computational and communication requirements of diverse applications in the LAE, introduces considerable challenges in integrating SAGIMEC into the LAE. In this work, we first present a digital twin-assisted SAGIMEC paradigm for LAE, where digital twin enables reliable network monitoring and management, while SAGIMEC provides efficient computing offloading services for Internet of Things sensor devices (ISDs). Then, a \underline{\textbf{j}}oint \underline{\textbf{s}}atellite selection, \underline{\textbf{c}}omputation offloading, \underline{\textbf{c}}ommunication resource allocation, \underline{\textbf{c}}omputation resource allocation and UAV trajectory \underline{\textbf{c}}ontrol optimization problem ($\text{JSC}^4\text{OP}$) is formulated to maximize the quality of service (QoS) of ISDs. Given the complexity of $\text{JSC}^4\text{OP}$, we propose an online decentralized optimization approach (ODOA) to address the problem. Specifically, $\text{JSC}^4\text{OP}$ is first transformed into a real-time decision-making optimization problem (RDOP) by leveraging Lyapunov optimization. Then, to solve the RDOP, we introduce an online learning-based latency prediction method to predict the uncertain system environment and a game theoretic decision-making method to make real-time decisions. Finally, theoretical analysis confirms the effectiveness of the ODOA, while the simulation results demonstrate that the proposed ODOA outperforms other alternative approaches in terms of overall system performance.
\end{abstract}

}	

\maketitle
\IEEEdisplaynontitleabstractindextext
\IEEEpeerreviewmaketitle 

%
%
\section{Introduction}
\label{sec:Introduction}
\par \IEEEPARstart{A}{s} an emerging sector in modern economic development, the low-altitude economy (LAE) is increasingly becoming a key driver of regional economic growth and technological advancement. Specifically, by efficiently developing and utilizing low-altitude airspace resources (typically referring to altitudes below 1,000 meters), the LAE encompasses a wide range of applications, including environmental monitoring, aerial inspections, and low-altitude tourism~\cite{ye2024integrated}. Accordingly, several companies have already ventured into unprecedented commercial opportunities presented by the LAE, including Kespry in inspection and surveying, Da-Jiang in manufacturing aircrafts, and Amazon Prime in logistics and delivery~\cite{jiang20236g}.

\par With the rapid development of the LAE, unmanned and manned aircraft, such as unmanned aerial vehicles (UAVs) and electric vertical take-off and landing (eVTOL), have been widely deployed across various domains. In particular, UAV-assisted multi-access edge computing (MEC) is recognized as a highly promising application in the LAE. Specifically, Internet of Things sensor devices (ISDs) are experiencing explosive growth, driving the development of numerous intelligent applications such as autonomous driving, augmented reality (AR), and virtual reality (VR)~\cite{ZhouGLQ24}, thereby creating significant social and commercial value. However, a critical challenge lies in the fact that these intelligent applications often involve computationally intensive and latency-sensitive tasks, which conflict with the limited computational resources and energy capacity of ISDs~\cite{Mao2017}. In traditional terrestrial networking, terrestrial MEC has been proposed to provide energy-efficient and low-latency computational offloading services for resource-limited ISDs. However, terrestrial MEC heavily relies on ground-based infrastructure, which limits its effectiveness in remote areas lacking infrastructure or in disaster-stricken regions where infrastructure is damaged. In LAE networking, the UAVs equipped with communication and computing capabilities can effectively compensate for the shortcomings of terrestrial MEC, owing to high line-of-sight (LoS) probability, cost-effectiveness, and flexible mobility. Consequently, UAV-assisted MEC has gained significant attention and has become a focal point of research~\cite{PervezSYZ24,Sun2024,Sun2024a,Sun2024b}. Nevertheless, the limited onboard capabilities of UAVs, particularly their constrained computational resources and energy supplies, remain a critical bottleneck restricting the improvement of system performance~\cite{QuDWDWGW21}.

\par Thanks to the rapid deployment of mega-Low Earth Orbit (LEO) satellite constellations such as Starlink~\cite{mcdowell2020low}, OneWeb~\cite{radtke2017interactions}, and Kuiper~\cite{liu2020automatic}, space-air-ground integrated MEC (SAGIMEC) network is emerging as a promising architecture to provide seamless computation offloading services. Specifically, SAGIMEC is usually a three-tier computing architecture that integrates heterogeneous network components, including a terrestrial network, an aerial UAV network, and a space low earth orbit (LEO) satellite network~\cite{WangJYZNTJ23}. On the one hand, due to the wide coverage of LEO satellites and flexible mobility of UAVs, SAGIMEC greatly expands the application scenarios and coverage of edge computing. On the other hand, the low transmission latency and seamless connectivity of LEO satellites enable SAGIMEC to effectively combine cloud-edge computing resources to improve the resource utilization.

\par Nevertheless, fully exploring the benefits of the SAGIMEC network faces several fundamental challenges. \textbf{\textit{i) Computation Offloading.}} The heterogeneity of the networks in SAGIMEC leads to an uneven distribution of resources. Moreover, different ISDs usually have diverse computing requirements for resources~\cite{GaoYY24}. As a result, the heterogeneous resource distribution and ISD requirements lead to the complexity of computation offloading decisions. \textbf{\textit{ii) Satellite Selection.}} The dynamic topology of satellite networks leads to time-varying and uncertain satellite link conditions~\cite{Zhang2024}. Therefore, when multiple satellites are accessible, it is challenging to select appropriate satellites as relay nodes for the efficient use of cloud computing services based on satellite networks. \textbf{\textit{iii) Resource Management.}} Tasks of ISDs are often computation-hungry and latency-sensitive, imposing strict requirements on computing and communication resources. However, UAV networks usually have limited computing and spectrum resources. Therefore, the strict computing requirements make resource allocation difficult in resource-constrained UAV networks~\cite{Sun2024b}. \textbf{\textit{iv) Trajectory Control.}} While the mobility of UAVs enhances the elasticity and flexibility of MEC, the limited onboard battery capacity of UAVs leads to finite service time~\cite{NguyenLG24}, which makes it challenging to balance both the service time of UAVs and the QoS of ISDs. Furthermore, the complexity and dynamics of the SAGIMEC network pose significant challenges in achieving efficient network management to meet the robustness requirements of the LAE.

\par To address the aforementioned challenges, we incorporate digital twin technology into the SAGIMEC network. By constructing digital space models to evaluate the state information of physical entities within the network, digital twin technology facilitates real-time network monitoring and network management, while providing insights for decision-making. Moreover, we propose a novel online decentralized optimization approach (ODOA) for SAGIMEC networks that enables the joint optimization of computation offloading, satellite selection, communication resource allocation, computation resource allocation, and UAV trajectory control, to maximize the QoS of ISDs. Our main contributions can be briefly outlined as follows:
\begin{itemize}
\item \textbf{\textit{New MEC Paradigm for LAE.}} We propose a novel digital twin-assisted SAGIMEC paradigm for LAE. In this paradigm, a UAV and a cloud center are seamlessly connected via a satellite network to facilitate high-quality computing offload services. Meanwhile, the digital twin facilitates real-time network management and provides insights for decision-making through comprehensive network monitoring. Notably, our proposed architecture is inherently scalable, enabling a seamless transition to multi-UAV scenarios. Moreover, within this framework, we consider the time-varying computing requirements of ISDs, the resource and energy constraints of the UAV, as well as the dynamics and uncertainties of the satellite links to more accurately capture the real-world physical characteristics of the SAGIMEC network.

\item \textbf{\textit{QoS-Oriented Optimization Problem Formulation.}} We formulate a joint satellite selection, computation offloading, communication and computation resource allocation, and UAV trajectory control optimization problem ($\text{JSC}^4\text{OP}$) to maximize the QoS of ISDs. Additionally, we develop a QoS assessment model, which integrates task completion latency and ISD energy consumption. Moreover, we demonstrate that $\text{JSC}^4\text{OP}$ is inherently challenging to be solved directly due to its reliance on future information, the presence of uncertain network parameters, and its non-convex and NP-hard characteristics.

\item \textbf{\textit{Novel Online Approach Design.}} To solve the $\text{JSC}^4\text{OP}$, we propose an online decentralized optimization approach (ODOA). Specifically, we first transform the $\text{JSC}^4\text{OP}$ into a real-time decision-making optimization problem (RDOP) that only depends on current information by using the Lyapunov optimization. Then, for the RDOP, we propose an online learning-based latency prediction method to predict uncertain network parameters and a game theoretic decision-making method to make real-time decisions.
 
\item \textbf{\textit{Theoretical Analysis and Simulation Experiments.}} The effectiveness and performance of the designed ODOA are confirmed through theoretical analysis and simulation experiments. In particular, the theoretical analysis establishes that the ODOA not only satisfies the UAV energy consumption constraint, but also exhibits low computational complexity. Additionally, the simulation results demonstrate that the ODOA outperforms other alternative approaches in terms of overall system performance.
	
\end{itemize}

\par The subsequent sections of this work are structured as follows. In Section \ref{sec:Related Work}, an overview of the related work is provided. Section \ref{sec:System Model and problem Formulation} details the relevant system models. Section \ref{sec:problem Formulation} presents the problem formulation and analysis. The Lyapunov-based problem transformation is described in Section \ref{sec:Lyapunov-Based problem Transformation}. In Section \ref{sec:Algorithm Design}, the algorithm design and theoretical analysis are provided. Then, we demonstrate and discuss the simulation results in Section \ref{sec:Simulation Results}. Finally, this work is concluded in Section \ref{sec:Conclusion}.

%
%
\section{Related Work}
\label{sec:Related Work}
\par In this section, we provide a comprehensive review of the relevant studies pertaining to SAGIMEC network architecture, formulation of optimization problems, and optimization approaches. Furthermore, we emphasize the key distinctions between our work and the existing research.
\subsection{Space-Air-Ground Integrated Multi-Access Edge Computing Network}
\par As an emerging technology, SAGIMEC networks have attracted extensive attention and research. Various network architectures have been proposed to enhance the efficiency of computation offloading services. For example, in~\cite{GaoYY24,PaulSNPL24}, the authors investigated the SAGIMEC architecture consisting of multiple ground base stations, UAVs, and LEO satellites to meet computation-intensive requests from ground devices. Du \emph{et al.}~\cite{DuWSQZWN24} explored the architecture of single satellite and multiple UAVs collaboration to provide computation offloading services for Internet of Things devices. Huang \emph{et al.} \cite{HuangCXXHC24} conducted a study on the integration of hybrid multi-cloud services and MEC within SAGIMEC, focusing on the dynamic access capabilities of UAVs, multi-satellite access, and cloud service selection. Shen \emph{et al.}~\cite{ShenTWB24} introduced a slice-oriented task offloading framework for space-air-ground integrated vehicular networks, aiming to deliver differentiated quality-of-service guarantees for high-speed vehicles. Yu \emph{et al.}~\cite{YuGSWC22} proposed an space-air-ground integrated network framework to provide various Internet of Vehicles services for vehicles in remote areas.
\par However, existing studies still have certain limitations. \textbf{First}, the aforementioned studies usually assume that the satellite link status can be accurately measured, which may be impractical because of the high-speed mobility, long-distance transmission, and time-varying network topology of satellite networks. \textbf{Furthermore}, the integration of digital twin with SAGIMEC holds significant potential for application, yet research in this area remains relatively scarce. To this end, we propose a digital twin-assisted SAGIMEC paradigm that accounts for dynamic and uncertain satellite links, filling a gap in existing research.
\subsection{Formulation of Optimization Problems}
\par The formulation of the optimization problem is critical for enhancing the performance of the SAGIMEC networks. For example, Gao \emph{et al.} \cite{GaoYY24} conducted a study on task hosting, computing offloading, computing resource allocation, and association control in a SAGIMEC network, with the aim of minimizing the time-averaged network cost. Huang \emph{et al.} \cite{HuangCXXHC24} investigated a partial task offloading strategy and jointly optimized task offloading, task partitioning, UAV trajectory, and computing resource allocation to minimize system energy consumption and latency. To maximize the number of tasks meeting the delay constraints, Zhang \emph{et al.}~\cite{ZhangLHLXW23} formulated a joint optimization problem for offloading destinations and offloading quantities. Cheng \emph{et al.}~\cite{ChengLQZHSS19} proposed an online computation offloading method to learn optimal offloading decisions, thereby effectively reducing the overall system cost. Nguyen \emph{et al.}~\cite{NguyenLG24} formulated an optimization problem that jointly considers user scheduling, computation resource and bandwidth allocation, bit allocation, partial offloading control, and UAV trajectory control, aiming to minimize system energy consumption.
\par This work differs from the aforementioned research in the following aspects. \textbf{First}, we incorporate a more comprehensive set of optimization metrics, including satellite selection, task offloading, communication and computing resource allocation, and UAV trajectory control, to fully exploit the benefit of the SAGIMEC network. \textbf{Moreover}, the aforementioned studies primarily focus on the overall system performance, whereas this work emphasizes the performance of terminal devices.

\subsection{Optimization Approaches}
\par To address complex optimization problems, numerous studies have focused on designing efficient optimization approaches. Some of these studies have proposed high-performance offline approaches, which typically assume that the edge computing scenarios is static or that the computational demands of terminal devices are known in advance~\cite{MaoLHAY24,ChenZWHS25,Xu2021,UAV-H,Hu2019}. However, in many edge computing scenarios, such as the metaverse or real-time video analytics, computational demands arrive in a stochastic manner. Additionally, satellites and UAVs in SAGIMEC networks exhibit inherent dynamic mobility. Therefore, online approaches are essential for SAGIMEC networks to make real-time decisions without knowledge of future information.
\par Several studies have also explored online approaches. For example, Zhou \emph{et al.}~\cite{zhou2022two} developed a two time-scale online approach for caching and task offloading by leveraging the Lyapunov optimization framework. Miao \emph{et al.}~\cite{10388042} proposed a deep deterministic policy gradient (DDPG)-based algorithm to optimize the computational resources allocation and UAV flight trajectory for UAV-assisted MEC. To minimize the average power consumption of the system with randomly arriving user tasks, Hoang \emph{et al.}~\cite{10102429} developed a Lyapunov-guided deep reinforcement learning (DRL) framework. Cai \emph{et al.} \cite{Cai2025} proposed an online approach based on graph DRL to optimize task offloading and resource allocation decisions.
\par The Lyapunov-based optimization framework and DRL represent two viable methodologies for developing online approaches. While DRL is a powerful technique for training agents to make real-time decisions, it necessitates a substantial amount of sample data to learn optimal strategies and incurs significant computational overhead.  Therefore, we employ the Lyapunov-based optimization framework to devise our online approach. \textbf{In a departure from existing research}, we propose a novel decentralized  decision-making method based on game theory within the Lyapunov-based optimization framework. This proposed approach demonstrates both low computational complexity and superior performance.
%
%
\section{System Model}
\label{sec:System Model and problem Formulation}
\begin{figure}[t]
    \centering
    \setlength{\abovecaptionskip}{2pt}%
    \setlength{\belowcaptionskip}{2pt}%
    \includegraphics[width =3.5in]{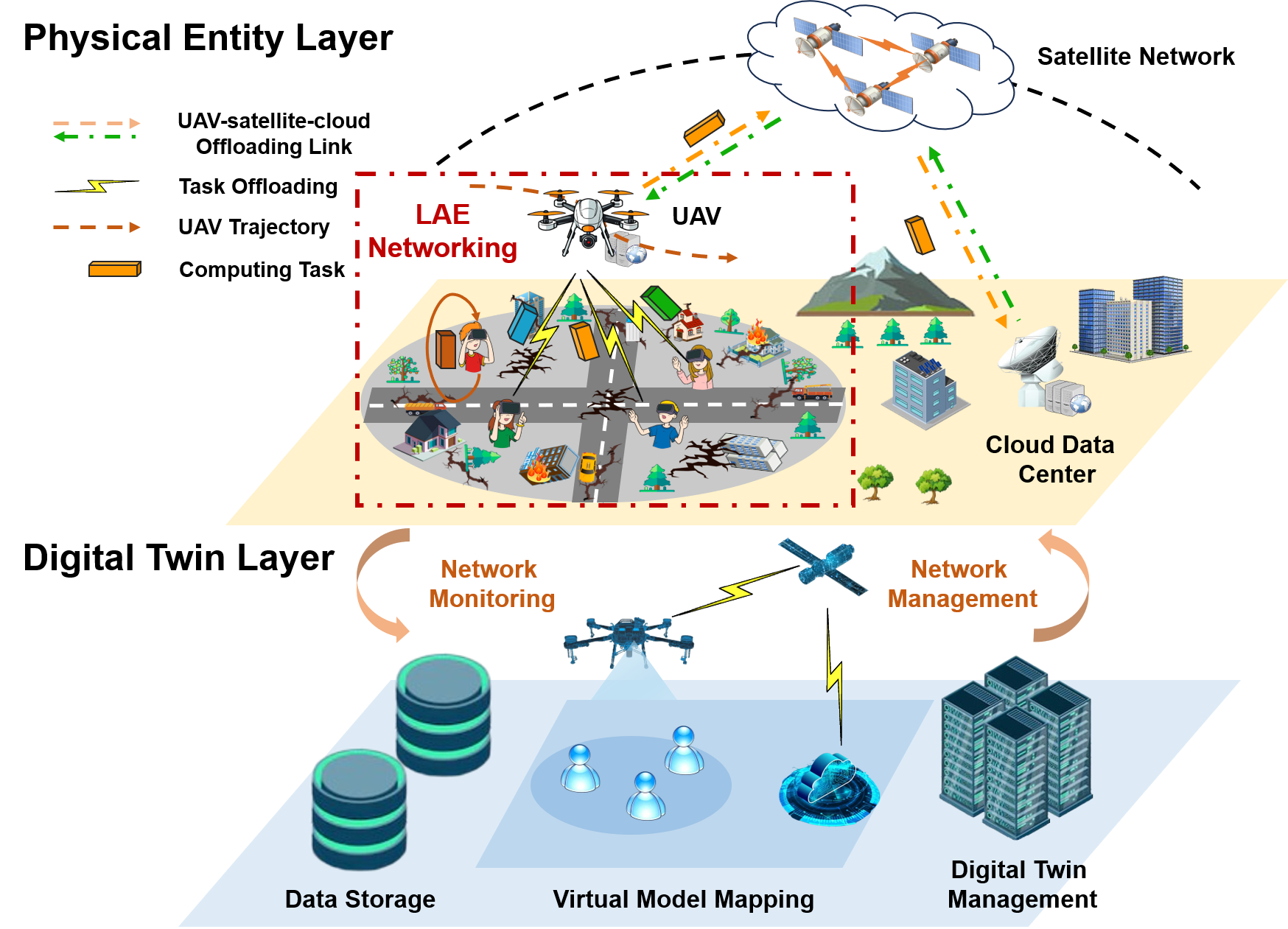}
    \caption{The proposed digital twin-assisted SAGIMEC architecture consists of a physical entity layer and a digital twin layer.} 
    \label{fig_gameModel}
\vspace{1em}
\end{figure}
\par As shown in Fig. \ref{fig_gameModel}, we propose a digital twin-driven SAGIMEC paradigm for the LAE, which consists of a physical entity layer and a digital twin layer, detailed as follows. 

\par \textbf{Physical Entity Layer.} In the spatial dimension, the physical entity layer is structured as a hierarchical architecture comprising a terrestrial layer, an aerial layer, and a space layer. \textit{At the terrestrial layer}, a set of ISDs $\mathcal{M}=\{1,\ldots,M\}$ are distributed in the considered area to carry out specific activities such as virtual/augmented reality applications, and generate corresponding tasks with time-varying computational demands. Moreover, a cloud center $c$ is deployed far away from the considered service area, which can provide robust cloud computing services. However, due to the destruction of ground infrastructure, e.g., during natural disasters, the ISDs are unable to access the cloud center via the terrestrial network. \textit{At the aerial layer}, a rotary-wing UAV $u$ equipped with computational and communication capabilities is employed in close proximity to the ISDs to provide flexible computing offloading services. \textit{At the space layer}, we consider a satellite constellation consisting of $N$ LEO satellites. This network provides seamless communication connectivity between the UAV and the cloud computing center, thereby making cloud computing services accessible. 

\par In the temporal dimension, the physical entity layer operates in a discrete time slot manner. Specifically, the system time is divided into $T$ equal time slots with $t\in \mathcal{T}=\{1,\ldots,T\}$, wherein each time slot has a duration of $\tau$~\cite{zhou2022two}. Furthermore, $\tau$ is chosen to be sufficiently small such that each time slot can be considered as quasi-static~\cite{Liu2022}.

\par \textbf{Digital Twin Layer.} The digital twin layer is the virtual representation of the physical entity layer, which is deployed in the cloud center to facilitate real-time monitoring and management of the physical network. Specifically, the digital twin layer consists of three key components, i.e., data storage, virtual model mapping, and digital twin management. First, the data storage is responsible for storing and collecting real-time information related to physical layer entities and network states. Second, the virtual model mapping involves creating virtual models of physical entities and simulating the interactions between these entities. Finally, the digital twin management is responsible for updating and maintaining the mapping between the virtual model and the physical network, while providing essential network control feedback to the physical network.

%
%
\subsection{Digital Twin Models}
\label{subsec:Basic Model}
\par The proposed system comprises three types of physical entities, i.e., ISD, UAV, and satellite. The cloud center creates corresponding digital twin models to record the real-time states of these physical entities, as detailed below.
\par \textbf{\textit{ISD Digital Twin Model.}} We consider that each ISD generates one computing task per time slot~\cite{JiangDXI23,ZhangHM23}. At time slot $t$, the digital twin model of ISD $m\in\mathcal{M}$ can be characterized as $\mathbf{St}_m(t)=\left(f_m^{\text{ISD}},E_m(t),\mathbf{\Phi}_m(t),\mathbf{q}_m\right)$, wherein $f_m^{\text{ISD}}$ and $E_m(t)$ denote the local computing capability and energy consumption of ISD $m$, respectively. Furthermore, the set $\mathbf{\Phi}_m(t)=\{D_m(t),\eta_m(t),T^{\text{max}}_m(t)\}$ represents the attributes of the computing task generated by ISD $m$, wherein $D_m(t)$ indicates the task data size (in bits), $\eta_m(t)$ denotes the computation density (in cycles/bit), and $T^{\text{max}}_m(t)$ is the deadline of the task (in s). Moreover, $\mathbf{q}_m=\left[x_m(t),y_m(t)\right]$ stands for the location coordinates of ISD $m$. 

\par \textbf{\textit{UAV Digital Twin Model.}} Similar to~\cite{li2023robust}, we consider that the UAV flies at a constant altitude $H$ to mitigate additional energy consumption associated with frequent altitude changes. Therefore, the digital twin model of the UAV $u$ can be characterized by $\mathbf{St}_{u}(t)=(\mathbf{q}_{u}^t,H,E_u(t),B_u,F_{u}^{\text{max}})$, wherein $\mathbf{q}_u^t = [x_u^t,y_u^t]$ and $H$ represent the horizontal coordinates and flight height, respectively. $E_u(t)$ signifies the energy consumption of the UAV. Moreover, $B_u$ and $F_{u}^{\text{max}}$ denote the available bandwidth resources and computing resources of the UAV, respectively.

\par \textbf{\textit{Satellite Digital Twin Model.}} The satellite digital twin model is developed to simulate the dynamic mobility characteristics of the satellite constellation. Specifically, due to the movement of satellites, the connectivity between the UAV and satellites varies over time, leading to a time-varying subset of accessible satellites $\mathcal{S}(t)$. Furthermore, due to the periodic nature of satellite movements, the concept of snapshots can be utilized to model the changes of the accessible subset~\cite{Zhang2024}. Specifically, every $\Delta$ consecutive time slots form a snapshot epoch, where the accessible subset remains constant within each snapshot epoch but varies across different snapshot epochs~\cite{song2022evolutionary}.
%
%
\subsection{Task Computing Model}
\label{subsec:Task Computing Model}
\par The task $\mathbf{\Phi}_m(t)$ generated by ISD $m$ can be carried out locally on the ISD (referred to as local computing), offloaded to UAV $u$ for execution (referred to as UAV-assisted computing), or offloaded to cloud $c$ for execution (referred to as cloud-assisted computing), which is decided by the offloading decision of the ISD. Therefore, we define a variable $a_m(t)\in \{l,u,c\}$ to indicate the offloading decision of ISD $m$ at time slot $t$, wherein $a_m(t)=l$ indicates the local computing, $a_m(t)=u$ represents the UAV-assisted computing, and $a_m(t)=c$ signifies the cloud-assisted computing. Furthermore, both local computing and computation offloading generally involve overheads in terms of latency and energy, which are explained in detail below.

\subsubsection{Local Computing} If task $\mathbf{\Phi}_m(t)$ is executed locally by ISD $m$ at time slot $t$, the ISD utilizes its local computational resources for task execution. 

\par \textbf{Task Completion Latency.} The task completion latency for local computing is computed as 
\begin{sequation}
    T_m^{\text{LC}}(t)=\frac{\eta_m(t)D_m(t)}{f_m^{\text{ISD}}},
\end{sequation}
\noindent where $f_m^{\text{ISD}}$ represents the computing capability of ISD $m$.

\par \textbf{ISD Energy Consumption.} The energy consumption of ISD $m$ to execute task $\mathbf{\Phi}_m(t)$ is calculated as 
\begin{sequation}
    E_m^{\text{LC}}(t)=k(f_m^{\text{ISD}})^{3}T_m^{\text{LC}}(t),
\end{sequation}

\noindent wherein $k$ is the effective switched capacitance coefficient dependent on the CPU chip architecture.~\cite{PanPWZW21}.

\subsubsection{UAV-assisted Computing} If task $\mathbf{\Phi}_m(t)$ is offloaded to UAV $u$ for execution at time slot $t$, the UAV establishes a communication connection with ISD $m$ to receive the task, allocates computing resources for task execution, and transmits the processed results to the ISD. Note that we ignore the latency and energy consumption associated with the result feedback due to the short-distance transmission and the small data size of the results~\cite{Zhang2024}. 

\par \textbf{ISD-UAV Communication.} The widely used orthogonal frequency-division multiple access (OFDMA) is employed to simultaneously serve multiple ISDs by the UAV~\cite{WeiCSNYZS21}. According to Shannon's formula, the data transmission rate from ISD $m$ to UAV $u$ at time slot $t$ is expressed as 
\begin{sequation}
    R_{u,m}(t)=w_{u,m}^t B_{u}\log_2\big(1+P_mg_{u,m}(t)/\varpi_0\big),
\end{sequation}

\noindent wherein $w_{u,m}^t$ represents the resource allocation coefficient of ISD $m$, $B_u$ denotes the bandwidth resources available to UAV $u$, $P_{m}$ indicates the transmission power of ISD $m$, $g_{u,m}(t)$ means the channel gain, and $\varpi_0$ is the noise power.

\par Considering that the ISD-UAV link may experience obstruction from environmental obstacles, we employ the probabilistic line-of-sight (LoS) channel model to describe the communication condition \cite{Zheng2024,communicationmodel2}. Specifically, the LoS probability $\rho_{u,m}^{\text{LoS}}(t)$ from UAV $u$ to ISD $m$ is defined as~\cite{Liu2024b} 
\begin{sequation}
    \rho_{u,m}^{\text{LoS}}(t)=\frac{1}{1+c_1 \exp(-c_2(\frac{180}{\pi} \arcsin{\frac{H}{d_{u,m}(t)}} - c_1))},
\end{sequation}

\noindent where $c_1$ and $c_2$ are the constants depending on the environment, and $d_{u,m}(t)$ means the straight-line distance between UAV $u$ and ISD $m$~\cite{sun2023uav}. Similar to~\cite{chen2023energy}, the path loss between UAV $u$ and ISD $m$ is evaluated as
\begin{sequation}
\begin{aligned}
    L_{u,m}(t)= & 20 \log _{10}\big(4 \pi f_u d_{u,m}(t)/v_c\big)+\rho_{u,m}^{\text{LoS}}(t) \eta^{\text{LoS}} \\
    & +\big(1-\rho_{u,m}^{\text{LoS}}(t)\big) \eta^{\text{nLoS}},
\end{aligned}
\end{sequation}

\noindent where $v_c$ denotes the speed of light, and $f_u$ represents the carrier frequency. $\eta^{\text{nLoS}}$ and $\eta^{\text{LoS}}$ denote the extra losses for nLoS and LoS links, respectively. Generally, the channel gain $g_{u,m}(t)$ is estimated as $g_{u,m}(t)=10^{-L_{u,m}/10}$.

\par \textbf{Task Completion Latency.} The task completion latency mainly consists of the transmission and execution latency, which is calculated as 
\begin{sequation}
    \label{eq.ec-delay}
    T_{m}^{\text{UC}}(t) = \frac{D_m(t)}{R_{u,m}(t)}+\frac{\eta_m(t) D_m(t)}{F_{u,m}^t},
\end{sequation}

\noindent where $F_{u,m}^t$ stands for the computational resources assigned by UAV $u$ to ISD $m$ at time slot $t$.

\par \textbf{ISD Energy Consumption.} The transmission energy consumption of ISD for UAV-assisted computing is computed as
\begin{sequation}
    \label{eq.trans-energy}
    E_{m}^{\text{UC}}(t) = \frac{P_mD_m(t)}{R_{u,m}(t)}.
\end{sequation}

\par \textbf{UAV Energy Consumption.} The computational energy consumption of UAV for UAV-assisted computing is given as
\begin{sequation}
    \label{eq.uav-comp-energy}
    E_{u,m}^{\text{comp}}(t) = \varpi\eta_m(t)D_m(t),
\end{sequation}

\noindent where $\varpi$ represents the energy consumption per unit CPU cycle of the UAV~\cite{JiangDXI23}. 

\par Therefore, at time slot $t$, the total computational energy consumption of the UAV is expressed as
\begin{sequation}
    \label{eq.comp-energy}
    E_u^{\text{comp}}(t) = \sum_{m\in \mathcal{M}}I_{\{a_m(t)=u\}}E_{u,m}^{\text{comp}}(t),
\end{sequation}

\noindent where $I_{\{X\}}$ represents an indicator function, which equals 1 if $X$ is true, and 0 otherwise. 

\subsubsection{Cloud-assisted Computing} If task $\mathbf{\Phi}_m(t)$ is offloaded to cloud $c$ for execution at time slot $t$, the task is first offloaded to the UAV. Then, the UAV further offloads the task to the cloud via the satellite network relay and receives the result feedback from the cloud. However, unlike ISD-UAV communication, the UAV-satellite-cloud channel is severely affected by long-distance transmission, unpredictable weather conditions, high-speed satellite mobility, and dynamic satellite network topology~\cite{niephaus2016qos}. Therefore, measuring the round-trip latency of tasks for UAV-satellite-cloud transmission accurately is challenging. Moreover, there are multiple accessible satellites $\mathcal{S}(t)$ per time slot. The selection of different satellite relays may result in varying round-trip latency. To this end, we introduce variables $b_u^t\in\mathcal{S}(t)$ and $L_s(t)$ to represent the satellite selection decision and the unit data round-trip latency for satellite $s\in \mathcal{S}(t)$, respectively. Specifically, the variable $L_s(t)$ ($L_s^{\text{min}}\leq L_s(t)\leq L_s^{\text{max}}$) stands for a random variable with an unknown mean, assumed to be independently and identically distributed across time slots~\cite{gao2021energy}. 

\par \textbf{Task Completion Latency.} The task completion latency mainly consists of the transmission latency from ISD $m$ to the UAV, the round-trip latency from the UAV to the cloud, and the cloud computing latency. Considering that cloud computing has sufficient computational resources, we ignore the corresponding computing delay. Therefore, the task completion latency is given as
\begin{sequation}
    \label{eq.delay_cc}
    T_m^{\text{CC}}(t) = \frac{D_m}{R_{u,m}(t)}+\sum_{s\in \mathcal{S}(t)} I_{\{b_u^t=s\}}D_m(t)L_s(t).
\end{sequation}

\par \textbf{ISD Energy Consumption.} For cloud-assisted computing, the ISD incurs transmission energy consumption, i.e.,
\begin{sequation}
    \label{eq.UD-energy_cc}
    E_m^{\text{CC}}(t) = \frac{P_mD_m(t)}{R_{u,m}(t)}.
\end{sequation}

\par \textbf{UAV Energy Consumption.} Similarly, the UAV also incurs transmission energy consumption for cloud-assisted computing, which is calculated as
\begin{sequation}
    \label{eq.UAV-energy_cc}
    E_{u,m}^{\text{trans}} = \sum_{s\in \mathcal{S}(t)} I_{\{b_u^t=s\}}D_m(t)Z_s(t),
\end{sequation}

\noindent where $Z_s(t)$ represents the energy consumption of transmitting per bit of data between the UAV and satellite $s$ at time slot $t$. Therefore, the total transmission energy consumption of the UAV is obtained as
\begin{sequation}
    \label{eq.total-UAV-energy_trans}
    E_u^{\text{trans}}(t) = \sum_{m\in \mathcal{M}}I_{\{a_m^t=c\}}E_{u,m}^{\text{trans}}(t).
\end{sequation}
%
%
\subsection{Performance Metrics}
\label{subsec:Performance Metrics}
\subsubsection{QoS of ISD} In this work, considering the latency sensitivity of computing tasks and the limited energy resources of ISDs, we evaluate the QoS of each ISD in each time slot by jointly considering the task completion latency cost and the ISD energy consumption cost. Specifically, the task completion latency of ISD $m$ at time slot $t$ is represented as
\begin{sequation}
    \label{eq.task-delay}
    T_m(t) = I_{\{a_m^t=l\}}T_m^{\text{LC}}+I_{\{a_m^t=u\}}T_m^{\text{UC}}+I_{\{a_m^t=c\}}T_m^{\text{CC}}.
\end{sequation}

\noindent Accordingly, at time slot $t$, the energy consumption of ISD $m$ is described as
\begin{sequation}
    \label{eq.GU-energy}
    E_m(t) = I_{\{a_m^t=l\}}E_m^{\text{LC}}+I_{\{a_m^t=u\}}E_m^{\text{UC}}+I_{\{a_m^t=c\}}E_m^{\text{CC}}.
\end{sequation}

\noindent Similar to~\cite{Chen2022,Ding2022}, at time slot $t$, the cost of ISD $m$ is formulated as
\begin{sequation}
    \label{eq.GU-cost}
    C_m(t)=\gamma^{\text{T}} T_m(t) + \gamma^{\text{E}}E_m(t),
\end{sequation}

\noindent where $\gamma^{\text{T}}$ and $\gamma^{\text{E}}$ (with $\gamma^{\text{T}}+\gamma^{\text{E}}=1$) denote the weight coefficients of latency and energy consumption, respectively. Clearly, minimizing the cost of ISDs is equivalent to maximizing the QoS of ISDs.

\subsubsection{UAV Energy Consumption} At time slot $t$, the total energy consumption of the UAV includes transmission energy consumption, computation energy consumption, and propulsion energy consumption. Similar to~\cite{pan2023joint,Zhang2024a,Huang2024a}, the propulsion power for a rotary-wing UAV with speed $v_u$ is given as
\begin{sequation}
\label{eq.prop-energy}
P_u(v_u)=\underbrace{C_1(1+\frac{3 v_u^2}{U_{\text{p}}^2})}_{\text {blade profile }}+\underbrace{C_2 \sqrt{\sqrt{C_3+\frac{v_u^4}{4}}-\frac{v_u^2}{2}}}_{\text{induced }}+\underbrace{C_4 v_u^3}_{\text {parasite}},
\end{sequation}

\noindent where $U_{\text{p}}$ refers to the rotor's tip speed, $C1$, $C2$, $C3$, and $C4$ are constants defined in~\cite{Yang2022}. Therefore, at time slot $t$, the total energy consumption of the UAV is calculated as
\begin{sequation}
    \label{eq.UAV-energy}
    E_u(t) = E_u^{\text{comp}}(t) + E_u^{\text{trans}}(t) +E_u^{\text{prop}}(t),
\end{sequation}

\noindent where $E_u^{\text{prop}}(t)=P_u(v_u(t))\tau$ denotes the propulsion energy consumption at time slot $t$. To guarantee service time, the UAV energy consumption constraint is defined as
\begin{sequation}
    \label{eq.eng-cons}
    \lim _{T \rightarrow+\infty} \frac{1}{T} \sum_{t=1}^{T} \mathbb{E}\left\{E_u(t)\right\} \leq \bar{E}_u,
\end{sequation}

\noindent where $\bar{E}_u$ is the energy budget of the UAV per time slot.

%
%
\section{Problem Formulation}
\label{sec:problem Formulation}
\par In this section, we formally formulate the optimization problem. Furthermore, we analyze the challenges of solving this problem and then explain the motivation behind the proposed approach.
\subsection{Optimization Problem}
\par To minimize the average costs of all ISDs over time (i.e., time-averaged ISD cost), this work jointly optimizes the computation offloading $\mathbf{A}=\{\mathcal{A}^t|\mathcal{A}^t=\{a_m^t\}_{m\in\mathcal{M}}\}_{t\in\mathcal{T}}$, satellite selection $\mathbf{B}=\{b_u^t\}_{t\in\mathcal{T}}$, computing resource allocation $\mathbf{F}=\{\mathcal{F}^t|\mathcal{F}^t=\{F_{u,m}^t\}_{m\in\mathcal{M}}\}_{t\in\mathcal{T}}$, communication resource allocation $\mathbf{W}=\{\mathcal{W}^t|\mathcal{W}^t=\{w_{u,m}^t\}_{m\in \mathcal{M}}\}_{t\in \mathcal{T}}$, and trajectory control $\mathbf{Q}=\{\mathbf{q}_u^t\}_{t\in \mathcal{T}}$. Mathematically, we can formulate this optimization problem as follows:

\vspace{-0.8em}
\begin{small}
 \begin{align}
    \textbf{P}: \quad &\underset{\mathbf{A}, \mathbf{B}, \mathbf{F}, \mathbf{W}, \mathbf{Q}}{\text{min}} \frac{1}{T}\sum_{t=1}^{T}\sum_{m=1}^{M}C_m(t) \label{P}\\
    \text{s.t.}\ \ 
    &\lim _{T \rightarrow+\infty} \frac{1}{T} \sum_{t=1}^{T} \mathbb{E}\left\{E_u(t)\right\} \leq \bar{E}_u,  \tag{\ref{P}{\text{a}}} \label{Pa}\\
    &a_m^t\in \{l,u,c\}, \forall m\in \mathcal{M}, t\in \mathcal{T}, \tag{\ref{P}{\text{b}}} \label{Pb}\\
    &I_{\{a_m^t=u\}}T^{\text{UC}}_m(t)+I_{\{a_m^t=c\}}T^{\text{CC}}_m(t)\leq T_m^{\text{max}}, \forall m\in \mathcal{M}, t\in \mathcal{T}, \tag{\ref{P}{\text{c}}} \label{Pc}\\
    &b_u^t\in \mathcal{S}(t), \forall t\in \mathcal{T},\tag{\ref{P}{\text{d}}} \label{Pd}\\
    &0< F_{u,m}^t \leq F_u^{\text{max}}, \forall m\in \mathcal{M}, t\in \mathcal{T}, \tag{\ref{P}{\text{e}}} \label{Pe}\\
    &\sum_{m=1}^M I_{\{a_m^t=u\}}F_{u,m}^t\leq F_u^{\text{max}}, \forall t\in \mathcal{T}, \tag{\ref{P}{\text{f}}} \label{Pf}\\
    &0< w_{u,m}^t \leq 1, \forall m\in \mathcal{M}, \forall t\in \mathcal{T}, \tag{\ref{P}{\text{g}}} \label{Pg}\\
    &\sum_{m=1}^M I_{\{a_m^t\in \{u,c\}\}}w_{u,m}^t\leq 1, \forall t\in \mathcal{T}, \tag{\ref{P}{\text{h}}} \label{Ph}\\
    &\mathbf{q}_u^{t=1}=\mathbf{q}^{\text{ini}}, \tag{\ref{P}{\text{i}}} \label{Pi}\\
    &||\mathbf{q}_u^{t+1}-\mathbf{q}_u^t|| \leq v_u^{\text{max}} \tau, \forall t\in \mathcal{T},\tag{\ref{P}{\text{j}}} \label{Pj}  
\end{align}
\end{small}

\noindent where $\mathbf{q}^{\text{ini}}$ represents the initial position of the UAV. Constraint (\ref{Pa}) is the long-term UAV energy consumption constraint. Constraints (\ref{Pb}) and (\ref{Pc}) pertain to computation offloading decisions. Constraint (\ref{Pd}) states that the UAV can only choose one satellite for communication. Constraints (\ref{Pe}) and (\ref{Pf}) regulate the allocation of computing resources, while constraints (\ref{Pg}) and (\ref{Ph}) govern the communication resource allocation. Finally, constraints (\ref{Pi}) and (\ref{Pj}) concern the UAV trajectory control.

\subsection{Problem Analysis}
\subsubsection{Challenges of Problem Solving} There are three main challenges in optimally solving the problem $\textbf{P}$. \textit{i) Future-dependent.} Obtaining the optimal solution of the problem requires complete future information, e.g., computing demands of all ISDs across all time slots. However, acquiring the future information is very challenging in the considered time-varying scenario. \textit{ii) Uncertainty.} Since the problem involves an uncertain network parameter, i.e., the task round-trip latency between the UAV and the cloud, it is challenging to make relevant decisions under uncertain network dynamics. \textit{iii) Non-convex and NP-hard.} The problem involves both continuous variables (i.e., resource allocation $\{\mathbf{F},\mathbf{W}\}$ and UAV trajectory control $\mathbf{Q}$) and discrete variables (i.e., computation offloading decision $\mathbf{A}$ and satellite selection decision $\mathbf{B}$). Therefore, it is a mixed-integer non-linear programming (MINLP) problem, which can be proven to be non-convex and NP-hard~\cite{boyd2004convex,belotti2013mixed}.

\subsubsection{Motivation of Proposing ODOA} Given the aforementioned challenges, while deep reinforcement learning (DRL) is often considered a viable approach, it may not be suitable for the considered system for the following reasons. First, DRL typically requires a large amount of training samples to learn effective strategies. However, it is difficult to obtain real sample data due to the time-varying and uncertain nature of the system. Furthermore, the formulated optimization problem involves numerous constraints, a high-dimensional decision space, and heterogeneous decision variables. Using DRL to solve this problem faces convergence challenges. Moreover, DRL is highly sensitive to environmental changes and lacks interpretability, which makes it difficult to meet the system's requirements for scalability and adaptability. 
\par Therefore, we proposed an efficient approach, i.e., ODOA, by leveraging Lyapunov optimization, online learning and game theory to transform and decompose the original problem based on the specific characteristics of the considered system and the optimization problem. Compared to DRL, the proposed approach does not require sample data and explicit knowledge of the system dynamics.  Additionally, the proposed approach offers interpretability and broader adaptability. Finally, the proposed approach has low computational complexity, making it suitable for real-time decision-making. Fig.~\ref{fig_algrithm_framework} shows the framework of the proposed ODOA. The details are elaborated in the following sections.
\begin{figure*}[!hbt]
	\centering
	\setlength{\abovecaptionskip}{0pt}%
	\setlength{\belowcaptionskip}{0pt}%
	\includegraphics[width =6.8in]{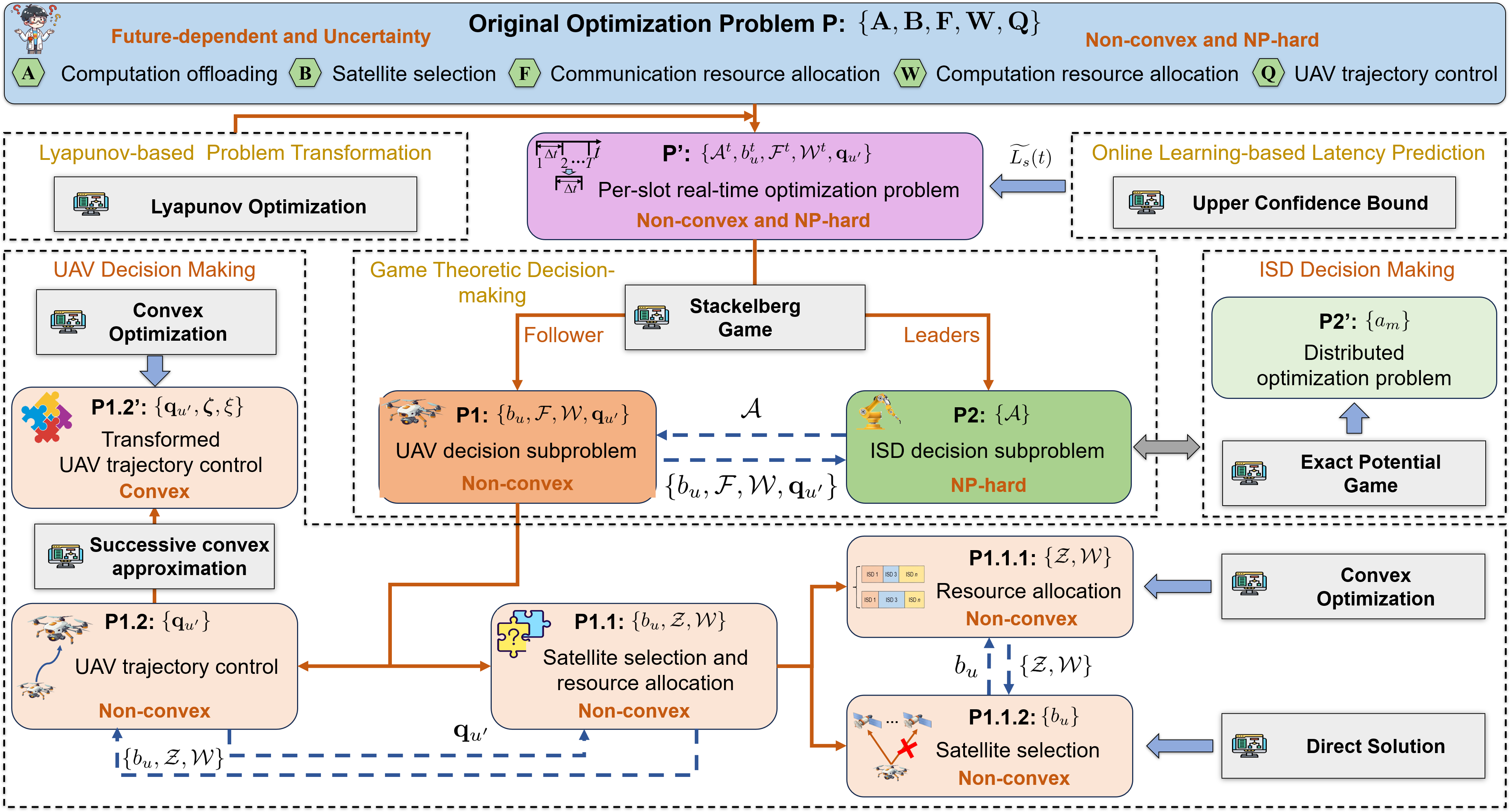}
	\caption{The framework of ODOA.} 
	\label{fig_algrithm_framework}
	\vspace{-1.5em}
\end{figure*}
%
%
\section{Lyapunov-Based Problem Transformation}
\label{sec:Lyapunov-Based problem Transformation}
\par Since problem $\textbf{P}$ is future-dependent, an online approach is necessary for real-time decision making without foreseeing the future. The Lyapunov optimization is a commonly used framework for online approach design, as it does not require direct knowledge of the network dynamics while providing guaranteed performance~\cite{Yang2022}. Therefore, we first utilize the Lyapunov optimization to transform problem $\textbf{P}$ into a per-slot real-time optimization problem.
\par Specifically, to meet the long-term UAV energy constraint (\ref{Pa}), the UAV digital twin model first defines two virtual energy queues $Q_{u1}(t)$ and $Q_{u2}(t)$, which represent the total transmission and computing energy queue, as well as the propulsion energy queue, respectively. Moreover, these queues are set as zero at the initial time slot, i.e., $Q_{u1}(1)=0$ and $Q_{u2}(1)=0$. Based on real-time monitoring of the UAV energy consumption, the virtual energy queues are updated as
\begin{sequation}
    \label{eq.queue}
    \begin{cases}
        Q_{u1}(t+1)=\max \left\{Q_{u1}(t)+E_{u1}(t)-\bar{E_{u1}}, 0\right\}, \\
        Q_{u2}(t+1)=\max \left\{Q_{u2}(t)+E_{u2}(t)-\bar{E_{u2}}, 0\right\},
    \end{cases}
\end{sequation}

\noindent where $E_{u1}(t)=E_u^{\text{comp}}(t)+E_u^{\text{trans}}(t)$ and $E_{u2}(t)=E_u^{\text{prop}}(t)$. Furthermore, $\bar{E_{u1}}$ and $\bar{E_{u2}}$ (with $\bar{E_{u1}}+\bar{E_{u2}}=\bar{E_u}$) represent the computing and transmission energy budgets per slot, as well as the propulsion energy budgets per slot, respectively. 

\par Secondly, the Lyapunov function, which represents a scalar measurement of the queue backlogs, is defined as 
\begin{sequation}
    \label{eq.ly-func}
    L(\boldsymbol{\Theta}(t)) = ((Q_{u1}(t))^2+(Q_{u2}(t))^2)/{2},
\end{sequation}

\noindent where $\boldsymbol{\Theta}(t)=[Q_{u1}(t),Q_{u2}(t)]$ is the vector of current queue backlogs. Thirdly, the one-slot conditional Lyapunov drift can be defined as
\begin{sequation}
    \label{eq.cond-ly-drift}
    \Delta L(\boldsymbol{\Theta}(t)) \triangleq \mathbb{E}\{L(\boldsymbol{\Theta}(t+1))-L(\boldsymbol{\Theta}(t)) \mid \boldsymbol{\Theta}(t)\}.
\end{sequation}

\noindent Finally, the \textit{drift-plus-penalty} can be given as
\begin{sequation}
    \label{eq.drift-plus-penalty}
   D(\boldsymbol{\Theta}(t)) =\Delta L(\boldsymbol{\Theta}(t))+V \mathbb{E}\left\{C_s(t)\mid \mathbf{\Theta}(t)\right\},
\end{sequation}

\noindent where $C_s(t)=\sum_{m=1}^{M}C_m(t)$ is the total cost of all ISDs at time slot $t$, and $V$ is a control parameter to trade off the total cost and the queue stability. Next, we provide an upper bound on the \textit{drift-plus-penalty}, as stated in Theorem~\ref{the:drift-plus-penalty}.

%
%
\begin{theorem}
\label{the:drift-plus-penalty}
For all $t$ and all possible queue backlogs $\boldsymbol{\Theta}(t)$, the drift-plus-penalty is upper bounded as
\begin{sequation}
    \label{eq.theorem1}
    \begin{aligned}
    D(\mathbf{\Theta}(t)) \leq & W + {Q}_{u1}(t)(E_{u1}(t)-\bar{E_{u1}})\\
    &+{Q}_{u2}(t)(E_{u2}(t)-\bar{E_{u2}})+V\times C_s(t),
    \end{aligned}
\end{sequation}

\noindent where $W=\frac{1}{2}\max \left\{\left(\bar{E_{u1}}\right)^2,\left(E_{u1}^{\max}-\bar{E_{u1}}\right)^2\right\}+\frac{1}{2} \max \left\{\left(\bar{E_{u2}}\right)^2,\left(E_{u2}^{\max}-\bar{E_{u2}}\right)^2\right\}$ is a finite constant.
\end{theorem}

%
%
\begin{proof}
From the inequality $(max\{a+b-c,0\})^2 \leq (a+b-c)^2\ \text{for} \ \forall a,b,c \geq 0$, we can know 
\begin{equation}
    \label{eq.a}
    Q_{u1}(t+1)^2 \leq (Q_{u1}(t)+E_{u1}(t)-\bar{E_{u1}})^2,
\end{equation}
\noindent Then, the following inequality holds,
\begin{sequation}
    \label{eq.temp1}
    \begin{aligned}
    (Q_{u1}(t+1)^2-Q_{u1}(t)^2)/{2} &\leq (E_{u1}(t)-\bar{E_{u1}})^2/{2}\\
    &+ Q_{u1}(t)(E_{u1}(t)-\bar{E_{u1}}),
    \end{aligned}
\end{sequation}

\noindent which also applies to queue $Q_{u2}(t)$. Therefore, we can obtain
\begin{sequation}
\label{eq.temp2}
\begin{aligned}
&\Delta L\left(\boldsymbol{\Theta}(t)\right) \leq \mathbb{E}\left\{\frac{\left(E_{u1}(t)-\bar{E_{u1}}\right)^2+(E_{u2}(t)-\bar{E_{u2}})^2}{2} \mid \boldsymbol{\Theta}(t)\right\}\\
& + \mathbb{E}\left\{Q_{u1}(t)(E_{u1}(t)-\bar{E_{u1}})+Q_{u2}(t)(E_{u2}(t)-\bar{E_{u2}}) \mid \boldsymbol{\Theta}(t)\right\} \\
&\leq W+Q_{u1}(t) \left(E_{u1}(t)-\bar{E_{u1}}\right)+Q_{u2}(t) \left(E_{u2}(t)-\bar{E_{u2}}\right),
\end{aligned}
\end{sequation}

\noindent where $E_{u1}^{\max}$ and $E_{u2}^{\max}$ are the upper bounds of computing and communication energy consumption, and propulsion energy consumption, respectively. Substituting \eqref{eq.temp2} into \eqref{eq.drift-plus-penalty}, we can prove the theorem. 
\end{proof}

\par According to the Lyapunov optimization framework, we minimize the right-hand side of inequality (\ref{eq.theorem1}). Therefore, the problem $\textbf{P}$ is converted into the problem $\textbf{P}^{'}$ that relies solely on current information to make real-time decisions, which is presented as follows:

\vspace{-0.8em}
\begin{small}
\begin{align}
\textbf{P}^{\prime}:&\underset{\mathcal{A}^t,b_u^t,\mathcal{F}^t,\mathcal{W}^t,\mathbf{q}_{u^\prime}}{\text{min}}Q_{u1}(t)E_{u1}(t)+Q_{u2}(t)E_{u2}(t)+V\sum_{m=1}^{M} C_m(t) \label{P_temp} \\
\text{s.t.} \ &(\ref{Pb})-(\ref{Pi}), \notag
\end{align}
\end{small}

\noindent where $\mathbf{q}_{u^\prime}=\mathbf{q}_{u}^{t+1}$ represents the UAV position at time slot $t+1$, and $V$ is control parameter. Note that although solving problem $\textbf{P}^{\prime}$ does not require future information, the problem still involves unknown network parameters, and it is an MINLP problem. Therefore, solving $\textbf{P}^{\prime}$ still remains challenging. To this end, in the following section, we design efficient algorithms to solve the problem.

%
%
\section{Algorithm Design}
\label{sec:Algorithm Design}

\par In this section, efficient algorithms are proposed to solve the problem $\textbf{P}^{\prime}$. Specifically, considering the uncertain task round-trip latency in $\textbf{P}^{\prime}$, we present an online learning-based latency prediction algorithm. Since $\textbf{P}^{\prime}$ is an MINLP problem, we further propose a game theoretic decision-making algorithm.

\subsection{Online Learning-based Latency Prediction Method}
\label{subsec:Online Learning-based Latency Prediction Algorithms}

\par Since the unit data round-trip latency $L_s(t)$ in problem $\textbf{P}^{\prime}$ is not precisely known before decision-making, online learning should be incorporated into the decision-making process to implicitly evaluate the statistics of the unit data round-trip latency based on network feedback. Specifically, if there are tasks requested for cloud computing services in each time slot, the UAV needs to choose a satellite $s$ from the available satellite set $\mathcal{S}(t)$ as a relay to transmit the tasks to the cloud. Then, the task processing results are relayed back to the UAV, and the corresponding unit data round-trip latency for satellite $s$ can be obtained. With the assistance of the digital twin, the information of the unit data round-trip delay can be stored in the digital twin in real time. By leveraging the accumulated historical information, the digital twin can predict the round-trip delay of different satellites as relays, thus providing insights for making high-quality offloading decisions. 
\par More specifically, the latency prediction can be modeled as a multi-armed bandit (MAB) problem~\cite{slivkins2019introduction} in the digital twin, wherein the UAV is considered as an agent and the satellite nodes in the available satellite set are treated as arms. However, a key issue for solving the MAB problem is the trade-off between exploitation and exploration. Specifically, when the UAV selects a satellite node, it can either exploit satellites with empirically lower round-trip latency to obtain better short-term benefits, or explore less frequently selected satellites to acquire new knowledge about their latency. Inspired by~\cite{li2019combinatorial}, we utilize the upper confidence bound (UCB) method to balance the trade-off between exploitation and exploration. Specifically, the prediction model for the unit data round-trip latency is presented as follows:
\begin{sequation}
\label{eq.latency-prediction}
\widetilde{L_s}(t)=
    \begin{cases}
        L_s^{\text{min}}, \text{if}\ t=1\ \text{or}\ h_s(t-1)=0,\\
        \max \left\{\bar{L_s}(t-1)-\omega_0 \sqrt{\frac{3 \log (\Delta_s(t))}{2h_s(t-1)}}, L_s^{\text{min}}\right\},\text{otherwise},
    \end{cases}
\end{sequation}

\noindent where $\omega_0 = L_s^{\text{max}}-L_s^{\text{min}}$, $\Delta_s(t)=\sum_{k=1}^tI_{\{s\in \mathcal{S}(k)\}}$ represents the number of time slots during which satellite $s$ is an accessible satellite until time slot $t$, $h_s(t-1)=\sum_{k=1}^{t-1}I_{\{b_u^t=s\}}$ denotes the number of time slots during which satellite $s$ was selected before time slot $t$, and $\bar{L_s}(t-1)=\sum_{k=1}^{t-1}I_{\{b_u^t=s\}}L_s(k)/h_s(t-1)$ is the observed average of the unit data round-trip latency for satellite $s$ based on collected historical feedback information. In this prediction model, $\omega_0 \sqrt{3 \log (\Delta_s(t))/(2h_s(t-1))}$ and $\bar{L_s}(t-1)$ are associated with exploration and exploitation, respectively. 
\par Based on the predicted unit data round-trip latency, the real-time decisions can be made by solving the problem $\textbf{P}^{\prime}$. Next, we introduce the proposed decision-making algorithm. Similar to~\cite{Cui2023}, we omit the time index for variables for the convenience of the following description.

\subsection{Game Theoretic Decision-making Method}
\label{subsec:Game Theory-based Decision-making Algorithm}
\par Since problem $\textbf{P}^{\prime}$ is non-convex and NP-hard, a centralized algorithm could impose significant computational overhead on decision-making, which might not be appropriate for the considered real-time decision-making scenario. Therefore, we propose a game-theoretic decentralized algorithm for real-time decision making. Specifically, ISDs make offloading decisions (i.e., $\mathcal{A}$) to obtain a better QoS, while the UAV as service provider devises the corresponding decisions (i.e., $\{b_u,\mathcal{F},\mathcal{W},\mathbf{q}_{u^\prime}\}$) with the aim of maximizing the QoS of all ISDs. Therefore, the problem $\textbf{P}^{\prime}$ can be treated as a multiple-leader common-follower Stackelberg game, where the UAV is the follower and ISDs are leaders. Based on the Stackelberg game, the problem $\textbf{P}^{\prime}$ can be decomposed into two subproblems, i.e., UAV decision subproblem and ISD decision subproblem, which are detailed as follows.

\par \textbf{UAV decision subproblem.} Given an arbitrary offloading decision $\mathcal{A}$ of ISDs, the problem $\textbf{P}^{\prime}$ can be converted into the subproblem $\textbf{P1}$ to make the relevant decisions of the UAV, which can be given as

\vspace{-0.8em}
{\small
\begin{align}
    &\textbf{P1}:\underset{b_u,\mathcal{F},\mathcal{W},\mathbf{q}_{u^\prime}}{\text{min}}Q_{u2}P_u(v_u)\tau+\sum_{m\in\mathcal{M}_{l}(\mathcal{A})}V\left(\gamma^{\text{T}}T_m^{\text{LC}}+\gamma^{\text{E}}E_m^{\text{LC}}\right)+ \notag\\    &\sum_{m\in\mathcal{M}_{u}(\mathcal{A})}\left[V\left(\gamma^{\text{T}}T_m^{\text{UC}}+\gamma^{\text{E}}E_m^{\text{UC}}\right)+Q_{u1} \varpi \eta_m D_m\right]+ \notag\\
    &\sum_{m\in\mathcal{M}_{c}(\mathcal{A})}\sum_{s\in \mathcal{S}}I_{\{b_u=s\}}\left[V\left(\gamma^{\text{T}}T_m^{\text{CC}}+\gamma^{\text{E}}E_m^{\text{CC}}\right)+Q_{u1} D_m Z_s\right]\label{P1}\\
    &\text{s.t.}\ \ (\ref{Pd})-(\ref{Pj}), \notag
\end{align}}

\noindent where $\mathcal{M}_{l}(\mathcal{A})$ represents the set of ISDs for local computing, $\mathcal{M}_{u}(\mathcal{A})$ denotes the set of ISDs for UAV-assisted computing, and $\mathcal{M}_{c}(\mathcal{A})$ is the set of ISDs for cloud-assisted computing, which is known based on the offloading decision $\mathcal{A}$ of ISDs.

\par \textbf{ISD decision subproblem.} For ISD $m$, let us define $U_m^{\text{LC}}$ as the utility of local computing, $U_m^{\text{UC}}$ as the utility of UAV-assisted computing, and $U_m^{\text{CC}}$ as the utility of cloud-assisted computing, which can be given as follows:

\vspace{-0.8em}
\begin{small}
\begin{align}
&U_m^{\text{LC}}=\gamma^{\text{T}}T_m^{\text{LC}}+\gamma^{\text{E}}E_m^{\text{LC}},\label{eq.u_loc}\\
&U_m^{\text{UC}}=Q_{u1}E_{u,m}^{\text{comp}}/V+\gamma^{\text{T}}T_m^{\text{UC}}+\gamma^{\text{E}}E_m^{\text{UC}},\label{eq.u_uc}\\
&U_m^{\text{CC}}=Q_{u1}E_{u,m}^{\text{trans}}/V+\gamma^{\text{T}}T_m^{\text{CC}}+\gamma^{\text{E}}E_m^{\text{CC}}.\label{eq.u_cc}
\end{align}
\end{small}

\noindent Therefore, the utility of ISD $m$ can be expressed as
\begin{sequation}
   \label{eq.utility}
   U_m(\mathcal{A}) =
       \begin{cases}
           U_m^{\text{LC}},a_m=l,\\
           U_m^{\text{UC}},a_m=u,\\
           U_m^{\text{CC}},a_m=c.
       \end{cases} 
\end{sequation}

\par According to the announced decisions $\{b_u,\mathcal{F},\mathcal{W},\mathbf{q}_{u^\prime}\}$ of the UAV and removing irrelevant constant terms, the problem $\textbf{P}^{\prime}$ can be transformed into the problem $\textbf{P2}$ to make the offloading decisions of ISDs, which can be given as

\vspace{-0.8em}
{\small \begin{align}
    \textbf{P2}:\quad &\underset{\mathcal{A}}{\text{min}}\ V\sum_{m\in \mathcal{M}}U_m(\mathcal{A}) \label{P2}\\
    \text{s.t.}\ \
    &(\ref{Pb})\ \text{and}\ (\ref{Pc}), \notag
\end{align}}

\noindent where each ISD seeks to optimize its utility by selecting an appropriate offloading strategy.

\subsubsection{UAV decision making} To solve the problem $\textbf{P1}$, assuming a feasible UAV position $\mathbf{q}_{u^{\prime}}$, we first optimize the satellite selection $b_u$ and resource allocation $\{\mathcal{F},\mathcal{W}\}$. Then, based on the obtained $b_u^*$ and $\{\mathcal{F}^*,\mathcal{W}^*\}$, we optimize the UAV position $\mathbf{q}_{u^{\prime}}$. The details are described as follows.

\par \textbf{Satellite selection and resource allocation.} Assuming a feasible $\mathbf{q}_{u^\prime}$, $\textbf{P1}$ can be transformed into a satellite selection and resource allocation subproblem $\textbf{P1.1}$. Defining $z_{u,m}=F_{u,m}/F_u^{\text{max}}$ and removing irrelevant constant terms, the subproblem can be be formulated as

\vspace{-0.8em}
{\small
\begin{align}
    &\textbf{P1.1}:\underset{b_u,\mathcal{Z},\mathcal{W}}{\text{min}}\sum_{m\in\mathcal{M}_{u}(\mathcal{A})}V\left[\gamma^{\text{T}}\left(\frac{D_m}{w_{u,m}r_{u,m}}+\frac{D_m\eta_m}{z_{u,m}F_u^{\text{max}}}\right)+\right. \notag\\
    &\left.\frac{\gamma^{\text{E}}P_mD_m}{w_{u,m}r_{u,m}}\right]+\sum_{m\in\mathcal{M}_{c}(\mathcal{A})}\sum_{s\in \mathcal{S}}I_{\{b_u=s\}}\left\{V\left[\frac{\gamma^{\text{E}}P_mD_m}{w_{u,m}r_{u,m}}+\right.\right.\notag\\
    &\left.\left.\gamma^{\text{T}}\left(\frac{D_m}{w_{u,m}r_{u,m}}+D_mL_s\right)\right]+Q_{u1} D_m Z_s\right\}\label{P1.1}\\
    &\text{s.t.}\ \ (\ref{Pd})-(\ref{Ph}), \notag
\end{align}}

\noindent where $r_{u,m}=B_{u}\log_2\big(1+\frac{P_mg_{u,m}(t)}{\varpi_0}\big)$, and $\mathcal{Z}=\{z_{u,m}\}_{m\in\mathcal{M}_u(\mathcal{A})}$. Then, by solving the problem \textbf{P1.1}, we can obtain the closed-form optimal solutions for resource allocation and satellite selection, which are described in Theorems \ref{the:resouce-allocation} and \ref{the:satellite-selection}, respectively.

\begin{theorem}
\label{the:resouce-allocation}
The optimal resource allocation can be given as follows:
\begin{sequation}
	\label{eq.resource-allocation-solution}
	\begin{cases}
            z_{u,m}^{*} = \frac{\sqrt{\gamma^{\mathrm{T}}\eta_mD_m/F_u^{\mathrm{max}}}}{\sum_{i\in \mathcal{M}_u(\mathcal{A})}\sqrt{\gamma^{\mathrm{T}}\eta_iD_i/F_u^{\mathrm{max}}}},\\
            w_{u,m}^{*} = \frac{\sqrt{(\gamma^{\mathrm{T}}D_m+\gamma^{\mathrm{E}}P_mD_m)/r_{u,m}}}{\sum_{i\in \mathcal{M}_o(\mathcal{A})}\sqrt{(\gamma^{\mathrm{T}}D_i+\gamma^{\mathrm{E}}P_iD_i)/r_{u,i}}},
	\end{cases}
\end{sequation}
\end{theorem}
\noindent where $\mathcal{M}_o(\mathcal{A}) = \mathcal{M}_u(\mathcal{A})\cup \mathcal{M}_c(\mathcal{A})$ represents the set of ISDs who perform computation offloading.
\begin{proof}
Given an arbitrary satellite selection strategy $b_u$ and removing irrelevant constant terms, the problem $\textbf{P1.1}$ can be reformulated into a resource allocation subproblem $\textbf{P1.1.1}$

\vspace{-0.8em}
{\small
\begin{align}
    \textbf{P1.1.1}:\underset{\mathcal{Z},\mathcal{W}}{\text{min}}&\sum_{m\in\mathcal{M}_u(\mathcal{A})}\frac{\gamma^{\text{T}}\eta_mD_m}{z_{u,m}F_u^{\text{max}}}+\sum_{m\in\mathcal{M}_o(\mathcal{A})}\frac{\gamma^{\text{T}}D_m+\gamma^{\text{E}}P_mD_m}{w_{u,m}r_{u,m}} \label{P1.1-1}\\
    \text{s.t.}\ \
    & z_{u,m}> 0, \forall m\in \mathcal{M}_u(\mathcal{A}), \tag{\ref{P1.1-1}{a}} \label{P1.1-1a}\\
    & \sum_{m\in\mathcal{M}_u(\mathcal{A})} z_{u,m}\leq 1, \tag{\ref{P1.1-1}{b}} \label{P1.1-1b}\\
    & w_{u,m}> 0, \forall m\in \mathcal{M}_o(\mathcal{A}), \tag{\ref{P1.1-1}{c}} \label{P1.1-1c}\\
    & \sum_{m\in\mathcal{M}_o(\mathcal{A})} w_{u,m}\leq 1. \tag{\ref{P1.1-1}{d}} \label{P1.1-1d}
\end{align}}

\par First, we can prove that problem $\textbf{P1.1.1}$ is a convex optimization problem because the Hessian matrix of (\ref{P1.1-1}) is positive semidefinite within the domain specified by (\ref{P1.1-1a}) to (\ref{P1.1-1d}). Therefore, the optimal resource allocation can be obtained by Karush–Kuhn–Tucker (KKT) conditions~\cite{boyd2004convex}. 
\end{proof}

\begin{theorem}
\label{the:satellite-selection}
The optimal satellite selection can be given as
\begin{sequation}
    \label{eq.d_u-solution}
    b_u^*\in \mathcal{S}^{\mathrm{sel}}=\arg \min_{s \in \mathcal{S}}\left(V\gamma^{\mathrm{T}}L_s+Q_{u1} Z_s\right),
\end{sequation}

\noindent where $\mathcal{S}^{\mathrm{sel}}$ represents the candidate satellite set. Note that if $\mathcal{S}^{\mathrm{sel}}$ contains multiple satellite nodes, the UAV would randomly select one from $\mathcal{S}^{\mathrm{sel}}$.
\end{theorem}

\begin{proof}
By substituting the optimal resource allocation solutions (\ref{eq.resource-allocation-solution}) into the problem $\textbf{P1.1}$ and removing irrelevant constant terms, the problem $\textbf{P1.1}$ can be reformulated into the satellite selection decision problem $\textbf{P1.1.2}$, which can be presented as follows:

\vspace{-0.8em}
{\small
\begin{align}
    \textbf{P1.1.2}:\underset{b_u}{\text{min}}\ &\sum_{s\in\mathcal{S}}I_{\{b_u=s\}}\sum_{m\in\mathcal{M}_c(\mathcal{A})}\left(V\gamma^{\mathrm{T}}L_s+Q_{u1} Z_s\right)D_m\label{P1.1-2}\\
    \text{s.t.}\ \
    &b_u\in \mathcal{S}. \notag
\end{align}}

\noindent By solving this problem, we can determine the optimal decision for satellite selection, as defined in equation~\eqref{eq.d_u-solution}.
\end{proof}

\par \textbf{UAV trajectory Control.} Given the optimal satellite selection decision $b_u^{*}$, resource allocation $\{\mathcal{F}^{*},\mathcal{W}^{*}\}$, and removing irrelevant constant terms, the problem $\textbf{P1}$ can be converted into the subproblem $\textbf{P1.2}$ to decide the UAV trajectory control, i.e.,

\vspace{-0.8em}
{\small \begin{align}
    &\textbf{P1.2}:\ \underset{\mathbf{q}_{u^\prime}}{\text{min}}\ V \sum_{m\in \mathcal{M}_o(\mathcal{A})}\frac{\gamma^{\text{T}}D_m+\gamma^{\text{E}}P_mD_m}{w_{u,m}^{*}B_u\log_2(1+\frac{\phi_m}{\|\mathbf{q}_{u^\prime}-\mathbf{q}_m\|^2+H^2})}+\notag \\
    &Q_{u2}\left(C_1\big(1+\frac{3 v_u^2}{U_{\text {p}}^2}\big)+C_2\sqrt{\sqrt{C_3+\frac{v_u^4}{4}}-\frac{v_u^2}{2}}+C_4v_u^3\right)\tau \label{P1.2}\\
    &\text{s.t.}\ \ (\ref{Pi})-(\ref{Pj}), \notag
\end{align}}

\noindent where $\mathbf{q}_{u^\prime}=\mathbf{q}_{u}^{t+1}$, $\mathbf{q}_u=\mathbf{q}_{u}^t$, $v_n=||\mathbf{q}_{u^\prime}-\mathbf{q}_{u}||/\tau$, and $\phi_{m} = (P_m10^{-(20 \log_{10}(4 \pi f_u /v_c)+\rho_{u,m}^{\text{LoS}} \eta^{\text{LoS}}+(1-\rho_{u,m}^{\text{LoS}}) \eta^{\text{nLoS}})/10})/ \varpi_0$. Clearly, the function (\ref{P1.2}) is non-convex concerning $\mathbf{q}_{n^\prime}$ due to the following non-convex terms
\begin{sequation}
\label{eq.non-convex-terms}
    \begin{cases}
    TM_m=\frac{1}{\log_2\big(1+\frac{\phi_m}{\|\mathbf{q}_{u^\prime}-\mathbf{q}_m\|^2+H^2}\big)},\ \forall m\in\mathcal{M}_o(\mathcal{A}), \\
    TM_0=\sqrt{\sqrt{C_3+v_u^4/4}-v_u^2/2}.
    \end{cases}
\end{sequation}

\noindent Next, we transform the objective function into a convex function by introducing slack variables.

\par For the non-convex term $TM_0$, we introduce the slack variable $\xi$ such that $\xi=TM_0$ and add the following constraint
\begin{equation}
\label{eq.slack1}
    \xi \geq \sqrt{\sqrt{C_3+v_u^4/4}-v_u^2/2} \Longrightarrow C_3/\xi^2 \leq \xi^2+v_u^2.
\end{equation}

\par For the non-convex term $TM_m$, we introduce the slack variable $\zeta_m$ such that $1/\zeta_m=TM_m$ and add the following constraint
\begin{equation}
\label{eq.slack2}
    \zeta_m \leq \log _2\big(1+\frac{\phi_{m}}{H^2+||\mathbf{q}_{u^{\prime}}-\mathbf{q}_m||^2}\big),\forall m\in\mathcal{M}_o(\mathcal{A}).
\end{equation}

\par According to the abovementioned relaxation transformation, the problem $\textbf{P1.2}$ can be equivalently transformed as
\begin{align}
    \mathbf{P1.2}^{\prime}:&\ \underset{\mathbf{q}_{u^\prime},\boldsymbol{\zeta},\xi}{\text{min}}\ V \sum_{m\in \mathcal{M}_o(\mathcal{A})}\frac{\gamma^{\text{T}}D_m+\gamma^{\text{E}}P_mD_m}{w_{u,m}^{*}B_u\zeta_m}\notag\\
    &+Q_{u2}\left(C_1\left(1+3 v_u^2/U_{\text {p}}^2\right)+C_2\xi_n+C_4v_u^3\right)\tau\label{P1.2-temp}\\
    \text{s.t.} \ &(\ref{Pi})-(\ref{Pj}),(\ref{eq.slack1})\ and\ (\ref{eq.slack2}), \notag
\end{align}
\noindent where $\boldsymbol{\zeta}=\{\zeta_{m}\}_{m\in\mathcal{M}_o(\mathcal{A})}$. For problem $\mathbf{P1.2}^{\prime}$, the optimization objective (\ref{P1.2}) is convex but constraints (\ref{eq.slack1}) and (\ref{eq.slack2}) are still non-convex. Similar to~\cite{UAV-H}, the successive convex approximation (SCA) method can be adopted to handle the non-convexity of above constraints, which is demonstrated in the following Theorems~\ref{pro:pro5-2-1} and~\ref{pro:pro5-2-2}.

\begin{theorem}
\label{pro:pro5-2-1}
Let $f(\mathbf{q}_{u^\prime},\xi)=\xi^2+v_u^2$, and given a local point $\mathbf{q}^{(i)}_{u^\prime}$ at the $i$-th iteration, a global concave lower bound of $f(\mathbf{q}_{u^\prime},\xi)$ can be obtained as follows:
{\small \begin{align}
    \label{eq.pro5-2-1}
    f^{(i)}(\mathbf{q}_{u^\prime},\xi) \triangleq&\left(\xi^{(i)}\right)^2+2 \xi^{(i)}\left(\xi-\xi^{(i)}\right)+\|\mathbf{q}_{u^{\prime}}^{(i)}-\mathbf{q}_u\|^2/\tau^2\notag\\
    &+2/\tau^2(\mathbf{q}_{u^{\prime}}^{(i)}-\mathbf{q}_u)^T\left(\mathbf{q}_{u^{\prime}}-\mathbf{q}_u\right),
\end{align}}

\noindent where $\xi^{(i)}$ is defined as
\begin{sequation}
    \label{eq.y-l}
    \xi^{(i)}=\sqrt{\sqrt{C_3+\frac{||\mathbf{q}_{u^{\prime}}^{(i)}-\mathbf{q}_u||^4}{4 \tau^4}}-\frac{||\mathbf{q}_{u^{\prime}}^{(i)}-\mathbf{q}_u||^2}{2 \tau^2}}.
\end{sequation}
\end{theorem}
\begin{proof}
Since $f(\mathbf{q}_{u^\prime},\xi)$ is a convex quadratic form, the first-order Taylor expansion of $f(\mathbf{q}_{u^\prime},\xi)$ at local point $\mathbf{q}^{(i)}_{u^\prime}$ is a global concave lower bound.
\end{proof}

\begin{theorem}
\label{pro:pro5-2-2}
Let $g_m(\mathbf{q}_{u^\prime})=\log _2\big(1+\frac{\phi_m}{H^2+||\mathbf{q}_{u^{\prime}}-\mathbf{q}_m||^2}\big)$, and given a local point $\mathbf{q}^{(i)}_{u^\prime}$ at the $i$-th iteration, a global concave lower bound of $g_m(\mathbf{q}_{u^\prime})$ can be obtained as follows:
{\small \begin{align}
    \label{eq.taylor2}
    &g_m^{(i)}(\mathbf{q}_{u^\prime}) \triangleq \log _2\big(1+\frac{\phi_m}{H^2+||\mathbf{q}_{u^{\prime}}^{(i)}-\mathbf{q}_m||^2}\big)\notag \\
    & -\frac{\phi_m (\log_2 e)(||\mathbf{q}_{u^{\prime}}-\mathbf{q}_m||^2-||\mathbf{q}_{u^{\prime}}^{(i)}-\mathbf{q}_m||^2)}{(\phi_m+H^2+||\mathbf{q}_{u^{\prime}}^{(i)}-\mathbf{q}_m||^2)(H^2+||\mathbf{q}_{u^{\prime}}^{(i)}-\mathbf{q}_m||^2)}.
\end{align}}
\end{theorem}
\begin{proof}
 We first consider the function $f(x)=\log_2(1+\frac{\phi}{(H^2+x)})$, where $\phi > 0$, and $x\geq 0$. The first-derivative of $f(x)$ can be calculated as:
    \begin{equation}
        \label{eq.f_x_first}
        \frac{\partial f(x)}{\partial x}=-\frac{\phi (\log_2e)}{[\phi+(H^2+x)](H^2+x)}.
    \end{equation}
    Accordingly, the second-order derivative of $f(x)$ can be calculated as
    \begin{equation}
        \label{eq.f_x_second}
        \frac{\partial^2 f(x)}{\partial x^2}=\frac{\phi (\log_2e)[2(H^2+x)+\phi]}{[\phi(H^2+x)+(H^2+x)^{2}]^2}.
    \end{equation}
    
    Clearly, since $\frac{\partial^2 f(x)}{\partial x^2}>0$, $f(x)$ is a convex function. Therefore, the first-order Taylor expansion of $f(x)$ at a local point $x_0$ is a global concave under-estimator of $f(x)$, i.e., the following inequality holds for any $x$
    \begin{sequation}
        \label{ieq1}
        f(x)\geq \log_2(1+\frac{\phi}{H^2+x_0})-\frac{\phi (\log_2e)(x-x_0)}{(\phi+H^2+x_0)(H^2+x_0)}.
    \end{sequation}
    
    Substituting $\phi=\phi_{m}$, $x=\|\mathbf{q}_{u^{\prime}}-\mathbf{q}_m\|^2$, and $x_0=\|\mathbf{q}^{(l)}_{u^{\prime}}-\mathbf{q}_m\|^2$ into the inequality (\ref{ieq1}), we can prove the theorem.
\end{proof}
\par According to Theorems \ref{pro:pro5-2-1} and \ref{pro:pro5-2-2}, at the $i$-th iteration, constraints (\ref{eq.slack1}) and (\ref{eq.slack2}) can be approximated as

\vspace{-0.8em}
{\small \begin{align}
    \label{eq.cons1}
    &\frac{C_3}{\xi^2}\leq f^{(i)}(\mathbf{q}_{u^\prime},\xi),\\
    &\zeta_m\leq g_m^{(i)}(\mathbf{q}_{u^\prime}),
\end{align}}

\noindent which are convex. Therefore, the problem $\mathbf{P1.2}^{\prime}$ is converted into a convex optimization problem, which can be efficiently resolved by off-the-shelf optimization tools such as CVX~\cite{cvx}. 

\subsubsection{ISD decision making} To decide the offloading decisions of ISDs, we can model the problem $\mathbf{P2}$ as a multi-ISDs computation offloading game (MISD-TOG). 
\par \textbf{Game Formulation.} Specifically, the MISD-TOG can be defined as a triplet $\Gamma=\{\mathcal{M},\mathbb{A}, (U_m(\mathcal{A}))_{m\in \mathcal{M}}\}$. 
\begin{itemize}
\item $\mathcal{M}=\{1,2,\dots,M\}$ denotes the set of players, i.e., all ISDs.
\item $\mathbb{A}=\mathbf{A}_1\times\dots\times\mathbf{A}_M$ represents the strategy space, wherein $\mathbf{A}_m=\{l,u,c\}$ is the set of offloading strategies for player $m\ (m\in \mathcal{M})$, $a_m\in\mathbf{A}_m$ denotes the offloading decision of player $m$, and $\mathcal{A}=(a_1,\dots,a_M)\in \mathbb{A}$ denotes a strategy profile.
\item $(U_m(\mathcal{A}))_{m\in \mathcal{M}}$ is the utility function of player $m$ that assigns a real number to each strategy profile $\mathcal{A}$.
\end{itemize}
\noindent In the game, each player strives to minimize its utility by selecting an appropriate offloading strategy. Therefore, the MISD-TOG can be mathematically described by a distributed optimization problem, i.e.,
\begin{sequation}
    \label{eq.task-offloading}
    \textbf{P2}^{\prime}:\underset{a_m}{\text{min}}\ U_m(a_m,a_{-m}),\ \forall m \in \mathcal{M},
\end{sequation}

\noindent where $a_{-m}=(a_1,\dots,a_{m-1},a_{m+1},\dots,a_M)$ denotes the offloading decisions of the other players except player $m$.

\par \textbf{The Solution of MISD-TOG.} To determine the solution of MU-TOG, we begin by introducing the concept of Nash equilibrium. A Nash equilibrium stands for a state in which no player is motivated to change their current strategy unilaterally. Definition~\ref{def:def1} presents the formal definition.

\begin{definition}
\label{def:def1}
If and only if a strategy profile $\mathcal{A}^*=(a_1^*,\dots,a_M^*)$ satisfies the following condition, it is a Nash equilibrium of game $\Gamma$
\begin{sequation}
    U_m(a_m^*,a_{-m}^*)\leq U_m(a_m^{\prime},a_{-m}^*) \quad  \forall a_m^{\prime}\in \mathbf{A}_m, m \in \mathcal{M}.
\end{sequation}
\end{definition}

\par Next, we introduce an important framework called the exact potential game~\cite{potential} through Definitions~\ref{def:def2} and~\ref{def:def_FIP}, to analyze whether there is a Nash equilibrium for MISD-TOG and how to obtain a Nash equilibrium.

\begin{definition}
    \label{def:def2}
     If the game $\Gamma$ has a potential function $F(\mathcal{A})$ that satisfies the following condition, it can be regarded as an exact potential game.

     \vspace{-0.8em}
     {\small
    \begin{align}
    \label{PG-def}
        &U_m(a_m,a_{-m})-U_m(a_m^{\prime},a_{-m})=F(a_m,a_{-m})-F(a_m^{\prime},a_{-m}),\notag\\ 
        &\forall (a_m,a_{-m}),(a_m^{\prime},a_{-m})\in \mathbb{A},   
    \end{align}}

\end{definition}

\begin{definition}
    A Nash equilibrium and a finite improvement path (FIP) always exist for an exact potential game with finite strategy sets~\cite{potential,2016Potential}.
    \label{def:def_FIP}
\end{definition}

\par The FIP implies that a Nash equilibrium can be obtained in a finite number of iterations by any best-response correspondence. Specifically, the best-response correspondence can be formally defined as follows:

\begin{definition}
    For each player $m\in\mathcal{M}$, their best response correspondence corresponds to a set-valued mapping $\mathbf{B}_m(a_{-m})$: $\mathbf{A}_{-m}\longmapsto \mathbf{A}_{m}$ such that
\begin{sequation}
    \label{eq.bestresponse}
        \mathbf{B}_m(a_{-m})=\left\{a_m^* \mid a_m^* \in \underset{a_m \in \mathbf{A}_m}{\arg \max } U_m\left(a_m, a_{-m}\right)\right\} .
    \end{sequation}
    \label{def:def_bestresponse}
\end{definition}
\par Therefore, by demonstrating that the MISD-TOG is an exact potential game, we can obtain a Nash equilibrium solution for it. The proof for this is provided in Theorem~\ref{theorem-PG}.

\begin{theorem}
\label{theorem-PG}
\par The MISD-TOG is an exact potential game with the potential function as follows:

{\small
\begin{align}
\label{eq.PF}
  F(\mathcal{A})=&\sum_{i\in \mathcal{M}}I_{\{a_i=u\}}\big(Q_{u1}E_{u,i}^{\mathrm{comp}}/V+\phi_i\sum_{j\leq i}I_{\{a_j=u\}}\phi_j\big)+\notag \\
  &\sum_{i\in \mathcal{M}}I_{\{a_i=c\}}\sum_{s\in\mathcal{S}}I_{\{b_u^*=s\}}\left(Q_{u1}E_{u,i}^{\mathrm{trans}}/V+\gamma^{\mathrm{T}}D_iL_s\right)+ \notag \\
  &\sum_{i\in \mathcal{M}}I_{\{a_i\in \{u,c\}\}}\gamma_i\sum_{j\leq i}I_{\{a_j\in \{u,c\}\}}\gamma_j+\sum_{i\in \mathcal{M}}I_{\{a_i=l\}}U_i^{\mathrm{LC}},
\end{align}}

\noindent where $\phi_i=\sqrt{\frac{\gamma^{\mathrm{T}}\eta_iD_i}{F_u^{\mathrm{max}}}}$, and $\gamma_i=\sqrt{\frac{\gamma^{\mathrm{T}}D_i+\gamma^{\mathrm{E}}P_iD_i}{r_{u,i}}}$.

\end{theorem}

\begin{proof}
Substituting \eqref{eq.resource-allocation-solution} into \eqref{eq.u_uc}, we can obtain the utility of UAV-assisted computing as follows:

\begin{sequation}
\label{eq.suav_ec}
    \begin{aligned}
    &U_m^{\mathrm{UC}}(\mathcal{A})=\frac{Q_{u1}}{V}E_{u,m}^{\mathrm{comp}}+\sqrt{\frac{\gamma^{\mathrm{T}}\eta_mD_m}{F_u^{\mathrm{max}}}}\sum_{i\in \mathcal{M}}I_{\{a_i=u\}}\sqrt{\frac{\gamma^{\mathrm{T}}\eta_iD_i}{F_u^{\mathrm{max}}}}\\
    &+\sqrt{\frac{\gamma^{\mathrm{T}}D_m+\gamma^{\mathrm{E}}P_mD_m}{r_{u,m}}}\sum_{i\in \mathcal{M}}I_{\{a_i\in \{u,c\}\}}\sqrt{\frac{\gamma^{\mathrm{T}}D_i+\gamma^{\mathrm{E}}P_iD_i}{r_{u,i}}}
    \end{aligned}
\end{sequation}

\noindent Similarly, substituting \eqref{eq.resource-allocation-solution} and \eqref{eq.d_u-solution} into \eqref{eq.u_cc}, we can obtain the utility of cloud-assisted computing as follows:
\begin{sequation}
\label{eq.luav_ec}
    \begin{aligned}
    &U_m^{\mathrm{CC}}(\mathcal{A})=\sum_{s\in \mathcal{S}}I_{\{b_u^*=s\}}(Q_{u1}E_{u,m}^{\mathrm{trans}}/V+D_mL_s)+ \\
    &\sqrt{\frac{\gamma^{\mathrm{T}}D_m+\gamma^{\mathrm{E}}P_mD_m}{r_{u,m}}}\sum_{i\in \mathcal{M}}I_{\{a_i\in \{u,c\}\}}\sqrt{\frac{\gamma^{\mathrm{T}}D_i+\gamma^{\mathrm{E}}P_iD_i}{r_{u,i}}}
    \end{aligned}
\end{sequation}

\par Let $\phi_{i}=\sqrt{\frac{\gamma^{\mathrm{T}}\eta_iD_i}{F_u^{\mathrm{max}}}}$ and $\gamma_{i}=\sqrt{\frac{\gamma^{\mathrm{T}}D_i+\gamma^{\mathrm{E}}P_iD_i}{r_{u,i}}}$, where $i\in\mathcal{M}$. Furthermore, for an arbitrary ordering of ISDs, let us introduce the following notation:
\begin{sequation}
    \begin{array}{ll}
    \phi^{\leq m}(\mathcal{A})=\underset{{j \leq m}}\sum I_{\{a_j=u\}}\phi_{j}, & \phi^{>m}(\mathcal{A})=\underset{{j>m}}\sum I_{\{a_j=u\}} \phi_{j}, \\
    \gamma^{\leq m}(\mathcal{A})=\underset{{j \leq m}}\sum I_{\{a_j\in\{u,c\}\}}\gamma_{j}, & \gamma^{>m}(\mathcal{A})=\underset{{j>m}}\sum I_{\{a_j\in\{u,c\}\}} \gamma_{j},
    \end{array}
\notag
\end{sequation}

\par Suppose ISD $m$ updates its current decision $a_m$ to the decision $a_m^{\prime}$ that leads to a change in its utility function, i.e., $U_m(a_m,a_{-m})-U_m(a_m^{\prime},a_{-m})$. Based on the definition of the potential function, i.e., Definition~\ref{def:def2}, we demonstrate through the following three cases which also leads to an equal change in the potential function. 

\par \textit{Case 1}: Suppose that $a_m = l$ and $a_m^{\prime} = u$. According to~\eqref{eq.PF}, we can obtain the following conclusion
\begin{sequation}
    \begin{aligned}
    \label{eq.change-1}
    F&(a_m,a_{-m})-F(a_m^{\prime},a_{-m})\\
    =&\ -Q_{u1}E_{u,m}^{\mathrm{comp}}/V-\phi_{m}\phi^{\leq m}(a_m^{\prime},a_{-m})-\gamma_{m}\gamma^{\leq m}(a_m^{\prime},a_{-m})\\
    &\ -\phi_{m}\phi^{> m}(a_m^{\prime},a_{-m})-\gamma_{m}\gamma^{> m}(a_m^{\prime},a_{-m})+U_m^{\mathrm{LC}}\\
    =&\ U_m(a_m,a_{-m})-U_m(a_m^{\prime},a_{-m}).
    \end{aligned}
\end{sequation}

\par \textit{Case 2}: Suppose that $a_m = l$ and $a_m^{\prime} = c$. According to~\eqref{eq.PF}, we can obtain the following conclusion
\begin{sequation}
    \begin{aligned}
    \label{eq.change-2}
    F&(a_m,a_{-m})-F(a_m^{\prime},a_{-m})\\
    =&\ -\sum_{s\in \mathcal{S}}I_{\{b_u^*=s\}}(Q_{u1}E_{u,m}^{\mathrm{trans}}/V+D_mL_s)-\gamma_{m}\gamma^{\leq m}(a_m^{\prime},a_{-m})\\
    &\ -\gamma_{m}\gamma^{> m}(a_m^{\prime},a_{-m})+U_m^{\mathrm{LC}}\\
    =&\ U_m(a_m,a_{-m})-U_m(a_m^{\prime},a_{-m}).
    \end{aligned}
\end{sequation}

\par \textit{Case 3}: Suppose that $a_m = u$ and $a_m^{\prime} = c$. According to~\eqref{eq.PF}, we can obtain the following conclusion
\begin{sequation}
    \begin{aligned}
    \label{eq.change-3}
    F&(a_m,a_{-m})-F(a_m^{\prime},a_{-m})\\
    =&\ Q_{u1}E_{u,m}^{\mathrm{comp}}/V+\phi_{m}\phi^{\leq m}(a_m,a_{-m})+\phi_{m}\phi^{> m}(a_m,a_{-m})\\
    &\ -\sum_{s\in \mathcal{S}}I_{\{b_u^*=s\}}(Q_{u1}E_{u,m}^{\mathrm{trans}}/V+D_mL_s)\\
    =&\ U_m(a_m,a_{-m})-U_m(a_m^{\prime},a_{-m}).
    \end{aligned}
\end{sequation}

\par Therefore, we can conclude that the MISD-TOG is an exact potential game.
\end{proof}

\par Finally, let us explore the impact of the constraint (\ref{Pc}) on the game $\Gamma$. This constraint may render some strategy profiles in $\mathbb{A}$ becoming infeasible, and this leads to a new game $\Gamma^{\prime}=\{\mathcal{M},\mathbb{A}^{\prime}, (U_m(\mathcal{A}))_{m\in \mathcal{M}}\}$. Theorem~\ref{theo:theo_cons} demonstrates that the game $\Gamma^{\prime}$ is also an exact potential game.

%
%
\begin{theorem}
\label{theo:theo_cons}
$\Gamma^{\prime}$ possesses the same potential function as $\Gamma$, which also is an exact potential game.
\end{theorem}

%
%
\begin{proof}
Since $\Gamma$ is an exact potential game, the equality (\ref{PG-def}) holds. Therefore, it remains valid if we restrict $(a_m,a_{-m})$ and $(a^{\prime}_m,a_{-m})$ to the new strategy space $\mathbb{A}^{\prime}$, which is a subset of $\mathbb{A}$. This validates the theorem..
\end{proof}

\par The proposed game-theoretic algorithm enables decentralized decision-making through interactions between ISDs and between ISDs and the UAV, effectively reducing computational complexity. Furthermore, with the assistance of digital twins, the interaction process can be simulated within the digital twin environment, thereby significantly reducing the communication overhead caused by the interactions.

%
%
\subsection{Main Steps of ODOA and Performance Analysis}
\label{sec:Performance Analysis}

\par In this section, the main steps of ODOA are described in Algorithm \ref{Algorithm 3}, and the corresponding analysis is provided. 

\begin{algorithm}
    \label{Algorithm 3}
    \SetAlgoLined
    \textbf{Initialization:} 
    $TUC = 0$, $\mathbf{q}_u(0)$, $Q_{u1}(0)$, $Q_{u2}(0)$\;
    \For{$t=1$ to $t=T$}
    {
        Acquire the ISD information $\{\mathbf{St}_m^{\text{ISD}}(t)\}_{m\in \mathcal{M}}$\;
        Calculate $L_s(t)\ (s\in \mathcal{S}(t))$ from the digital twin based on Eq.~\eqref{eq.latency-prediction}\;
        Based on Eqs.~\eqref{eq.resource-allocation-solution} and~\eqref{eq.d_u-solution}, obtain $\mathcal{A}^{*}$ by utilizing the exact potential game\;
        Calculate $\{b_u^{*},\mathcal{F}^{*},\mathcal{W}^{*}\}$ based on Eqs.~\eqref{eq.resource-allocation-solution} and~\eqref{eq.d_u-solution}\;
        Obtain $\mathbf{q}_{u^\prime}^{*}$ by solving problem $\mathbf{P1.2}$\;
        All ISDs execute their tasks according to $\mathcal{A}^{*}$ and obtain the respective cost $C_m^{*}(t)$\;
        Obtain system cost $C_s(t)=\sum_{m=1}^MC_m^{*}(t)$\;
        $TUC=TUC+C_s(t)$\;
        Update the queues $Q_{u1}(t+1)$ and $Q_{u2}(t+1)$ in the digital twin based on Eq.~\eqref{eq.queue}\;
        Update $t=t+1$\;
    }
    $TUC=TUC/T$\;
    \Return{Time-averaged ISD cost $TUC$.}
    \caption{ODOA}
\end{algorithm}

%
%
\begin{theorem}
\label{theorem-energy}
    The proposed ODOA can satisfy the UAV energy constraint defined in (\ref{eq.eng-cons}).
\end{theorem}

\begin{proof}
   According to Theorem 4.13 of~\cite{2010Neely}, we can conclude that all virtual queues are rate stable. Therefore, we have
\begin{equation}
    \label{eq.stable}
    \lim _{T \rightarrow+\infty} \frac{Q_{u1}(T)+Q_{u2}(T)}{T} = 0\ \text{with probability 1},\ \forall n\in \mathcal{N}. 
\end{equation}

\par Using the sample path property (Lemma 2.1 of ~\cite{2010Neely}), we have
\begin{sequation}
    \label{eq.spp}
\frac{Q_{u1}(T)+Q_{u2}(T)}{T}-\frac{Q_{u1}(1)+Q_{u2}(1)}{T} \geq \frac{1}{T} \sum_{t=1}^T \left(E_u(t)-\bar{E}_u\right).
\end{sequation}

\par By taking the infinite limit on both sides of~\eqref{eq.spp} and given that $Q_{u1}(1)=0$, $Q_{u2}(1)=0$, we can prove that formula (\ref{eq.eng-cons}) holds. 
\end{proof}

%
%
\begin{theorem}
    \label{the:complexity}
    The proposed ODOA has a worst-case polynomial complexity per time slot, i.e., $\mathcal{O}(I_cM+M^{3.5}\log_2(\frac{1}{\varepsilon}))$, wherein $M$ indicates the number of ISDs, $I_c$ denotes the number of iterations needed for MISD-TOG to reach the Nash equilibrium, and $\varepsilon$ is the accuracy parameter for SCA in solving problem $\mathbf{P1.2}^{\prime}$.
\end{theorem}

\begin{proof}
In general, ODOA contains a latency prediction process and a decision making process. For latency prediction, the complexity of the proposed algorithm can be regarded as $\mathcal{O}(1)$. For decision making, the complexity of the proposed algorithm mainly consists of solving problems $\mathbf{P2}$ and $\mathbf{P1.2}$. On the one hand, based on the analysis of~\cite{2016Potential}, the computational complexity of solving problem $\mathbf{P2}$ can be calculated as $\mathcal{O}(I_cM)$. On the other hand, according to the analysis in~\cite{Wang2022}, the computational complexity of solving problem $\mathbf{P1.2}$ is $\mathcal{O}(M^{3.5}\log_2(\frac{1}{\varepsilon}))$. Therefore, the worst-case computational complexity of ODOA can be obtained as $\mathcal{O}(I_cM+M^{3.5}\log_2(\frac{1}{\varepsilon}))$.
\end{proof}

%
%
\section{Simulation Results} 
\label{sec:Simulation Results}
\par In this section, the performance of the designed ODOA is evaluated through simulation experiments.
\subsection{Simulation Setup}
\begin{table}
	\setlength{\abovecaptionskip}{-1em}%
	\setlength{\belowcaptionskip}{0pt}%
	\caption{Simulation Parameters}
	\label{parameters}
	\renewcommand*{\arraystretch}{1.15}
	\begin{center}
		\begin{tabular}{p{.06\textwidth}<{\centering}|p{.21\textwidth}|p{.13\textwidth}}
			\hline
			\hline
			\textbf{Symbol}&\textbf{Meaning}&\textbf{Value (Unit)}\\
			\hline
			$D_m$ &Task size &$[0.5,3]$ Mb\\
			\hline
			$\eta_m$ &Computation intensity of tasks &$[500,1500]$ cycles/bit \\
			\hline
			$T_m^{\text{max}}$ &Maximum tolerable delay of tasks &$1$ s\\
			\hline
			$\alpha$ &Memory level of velocity &$0.9$ \\
                \hline
			$\overline{\mathbf{v}}_m$ &the velocity of ISD $m$ &$1$ m/s~\cite{Yang2022} \\
                \hline
			$\sigma_m$ & The asymptotic standard deviation of velocity &$2$~\cite{Yang2022} \\
			\hline
			$F_n^{\text{max}}$ &Computation resources of SUAVs &$20$ GHz\\
			\hline
			$v_n^{\text{max}}$ &Maximum flight speed of SUAVs &$25$ m/s~\cite{Yang2022}\\   
   			\hline
			$d^{\text{min}}$ &Minimum safety distance &$10$ m\\
			\hline
			$F_u^{\text{max}}$ &Computation resources of LUAV &$30$ GHz\\
			\hline
			$B_s$ &Bandwidth of MEC server $s$ &$10$ MHz $(s=u)$, \newline $5$ MHz $(s\in \mathcal{N})$ \\
			\hline
			$p_m$ &Transmission power of ISD $m$  &$20$ dBm \\
			\hline
			$\varpi_0$ &Noise power &$-98$ dBm\\
			\hline
			$c_1$, $c_2$ &Parameters for LoS probability &10, 0.6~\cite{communicationmodel1}\\
			\hline
			$\eta^{\text{LoS}},\eta^{\text{nLoS}}$ &Additional losses for LoS and nLoS links &1.0 dB, 20 dB\\ 
			\hline
			$\kappa$ &CPU parameters &$10^{-28}$\\
			\hline
			$\varpi$  &Energy consumption per unit CPU cycle of SUAVs&$8.2\times10^{-27}$ J~\cite{JiangDXI23} \\
			\hline
			$C_1$, $C_2$, $C_3$, $C_4$ & UAV propulsion power consumption parameters&80, 22, 263.4, 0.0092~\cite{Yang2022} \\
                \hline
                $\bar{E}_n$ & Energy budget per time slot for SUAV $n$ & 220 J\\
			\hline
			$U_{\text{p}}$ &Tip speed of the rotor&120 m/s \\
			\hline
			$\gamma_m^{\text{T}}$, $\gamma_m^{\text{E}}$ &The weight coefficients of task completion delay and energy consumption for ISD $m$ & 0.7, 0.3\\
			\hline
		\end{tabular}
	\end{center}
\end{table}
\subsubsection{Scenario Setting} We consider a SAGIMEC network, where a satellite network, a UAV, and a cloud computing center collaborate to provide computing offload services to $20$ ISDs within a $600\times600\ \text{m}^2$ service area. Furthermore, each epoch lasts 300 time slots with duration $\tau=1\ \text{s}$.
\subsubsection{Parameter Setting} For the satellite network, the unit data round-trip latency $L_s(t)$ (in s/bit) for each satellite is generated from a truncated Gaussian distribution~\cite{gao2021energy} with a mean of $(L_s^{\text{min}}+L_s^{\text{max}})/2$, where $L_s^{\text{min}}\in(15\times10^{-8},20\times10^{-8})$, and $L_s^{\text{max}}\in(30\times10^{-8},35\times10^{-8})$. For the UAV, we set the initial position to $\mathbf{q}_u^{\text{ini}}=[0,0]\ \text{m}$, and the fixed altitude to $H = 100\ \text{m}$. For the ISDs, the computing capacity of each ISD is randomly taken from $\{1,1.5,2\}\  \text{GHz}$. The default values for the remaining parameters are listed in Table~\ref{parameters}.
\subsubsection{Performance Metrics} We evaluate the overall performance of the proposed approach based on the following performance metrics. \textit{1) Time-averaged ISD cost} $\frac{1}{T}\sum_{t=1}^{T}\sum_{m=1}^{M}C_m(t)$, which represents the average cumulative cost of all ISDs per unit time. \textit{2) Average task completion latency} $\frac{1}{T}\sum_{t=1}^{T}\frac{1}{M}\sum_{m=1}^{M}T_m(t)$, which indicates the average latency for completing a task. \textit{3) Time-averaged ISD energy consumption} $\frac{1}{T}\sum_{t=1}^{T}\sum_{m=1}^{M}E_m(t)$, which signifies the cumulative energy consumption of ISDs over the system timeline. \textit{4) Time-averaged UAV energy consumption} $\frac{1}{T}\sum_{t=1}^{T}E_u(t)$, which means the average energy consumption of each SUAV per unit time. 
\subsubsection{Comparative Approaches} To demonstrate the effectiveness of ODOA, we compares ODOA with the following approaches. \textit{i) UAV-assisted computing (UAC)}: Tasks are executed locally on ISDs or offloaded to the UAV. \textit{ii) Equal resource allocation (ERA)~\cite{Josilo}}: The UAV allocates computing and communication resources equally. \textit{iii) $\varepsilon$-greedy~\cite{VermorelM05}}: A $\varepsilon$-greedy-based algorithm is adopted to balance exploration and exploitation. \textit{iv) Only consider QoS (OCQ)~\cite{JiangDXI23}}: Ignoring the UAV energy consumption constraint.

%
%

\subsection{Evaluation Results}
\label{subsec:Evaluation results}
\begin{figure*}[!hbt] 
	\centering
	\setlength{\abovecaptionskip}{1pt}%
	\setlength{\belowcaptionskip}{1pt}%
	\subfigure[]
	{
		\begin{minipage}[t]{0.23\linewidth}
			\centering
			\includegraphics[scale=0.38]{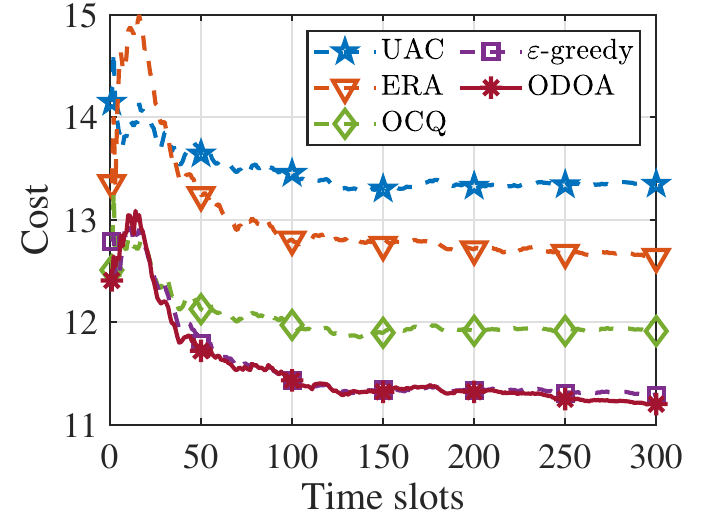}
		\end{minipage}
	}
	\subfigure[]
	{
		\begin{minipage}[t]{0.23\linewidth}
			\centering
			\includegraphics[scale=0.38]{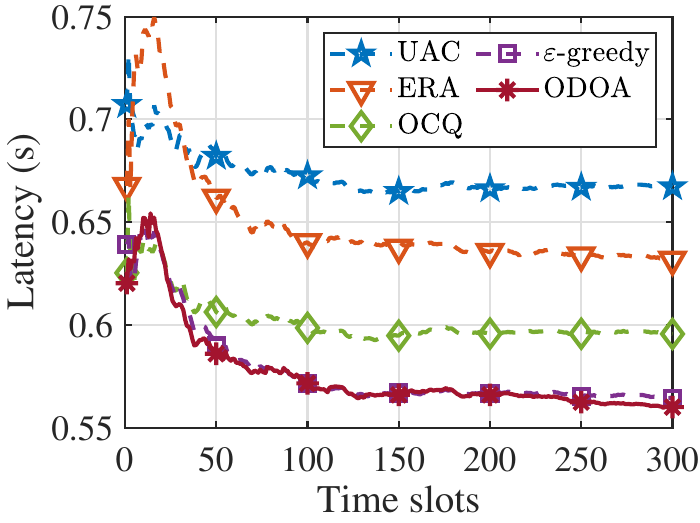}
		\end{minipage}
	}
	\subfigure[]
	{
		\begin{minipage}[t]{0.23\linewidth}
			\centering
			\includegraphics[scale=0.38]{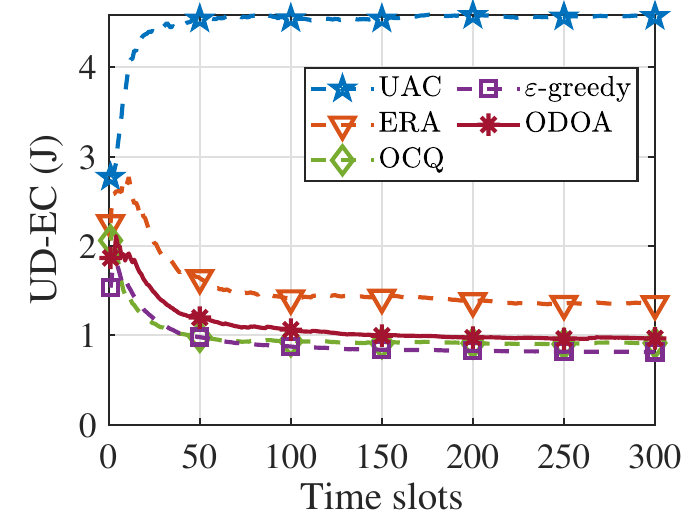}
		\end{minipage}
	}
        \subfigure[]
	{
		\begin{minipage}[t]{0.23\linewidth}
			\centering
			\includegraphics[scale=0.38]{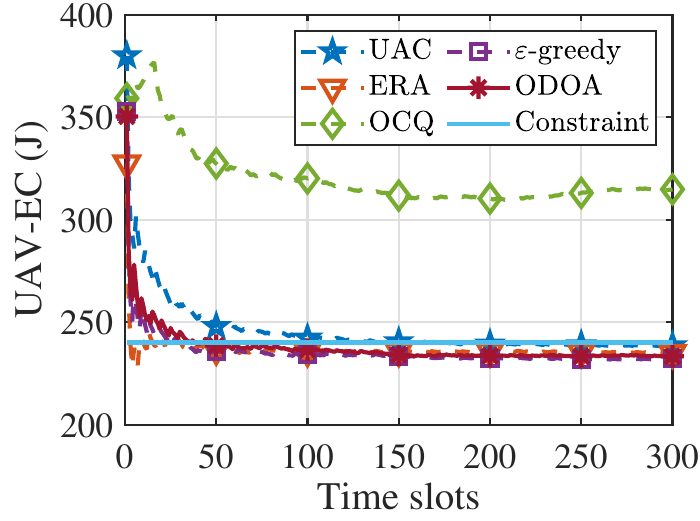}
		\end{minipage}
	}
	\caption{The impact of time slots on system performance. (a) Time-average ISD cost (Cost). (b) Average task completion latency (Latency). (c) Time-average ISD energy consumption (ISD-EC). (d) Time-average UAV energy consumption (UAV-EC).} 
	\label{fig_time}
	\vspace{-1em}
\end{figure*}
\subsubsection{Impact of Time} Figs.~\ref{fig_time}(a), \ref{fig_time}(b), \ref{fig_time}(c) and \ref{fig_time}(d) illustrate the dynamics of time-averaged ISD cost, average task completion latency, time-averaged ISD energy consumption among the five approaches, and time-averaged UAV energy consumption, respectively. It can be observed that the OCQ achieves worse performance compared the proposed ODOA in terms of time-averaged ISD cost, and average task completion latency. This is mainly due to the game-theoretic computation offloading algorithm. Specifically, regardless of the UAV energy constraint, more tasks are offloaded to the UAV, which leads to a decrease in the overall performance due to the limited computing resources of the UAV. Moreover, the proposed ODOA outperforms UAC, ERA, and $\varepsilon$-greedy with respect to time-averaged ISD cost and average task completion latency. The reasons are as follows. First, the three-tier architecture combined with cloud computing offers abundant resource supply. Second, the optimal resource allocation strategy can more effectively utilize the limited resources of the UAV. Third, the UCB-based algorithm better balances exploration and exploitation to improve the accuracy of latency prediction. Finally, as shown in Fig. \ref{fig_time}(d), the proposed ODOA can satisfy the long-term UAV energy constraint under the real-time guidance of the Lyapunov-based energy queue, which is consistent with the analysis in Theorem~\ref{theorem-energy}. 

\par In conclusion, the set of simulation results demonstrates the effectiveness of the ODOA in enhancing overall performance while adhering to the UAV energy constraint.
\begin{figure*}[!hbt] 
	\centering
	\setlength{\abovecaptionskip}{1pt}%
	\setlength{\belowcaptionskip}{1pt}%
	\subfigure[]
	{
		\begin{minipage}[t]{0.23\linewidth}
			\centering
			\includegraphics[scale=0.38]{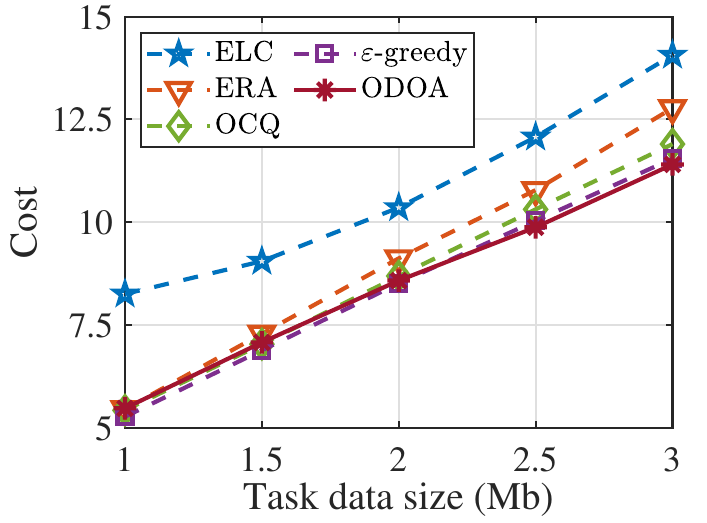}
		\end{minipage}
	}
	\subfigure[]
	{
		\begin{minipage}[t]{0.23\linewidth}
			\centering
			\includegraphics[scale=0.38]{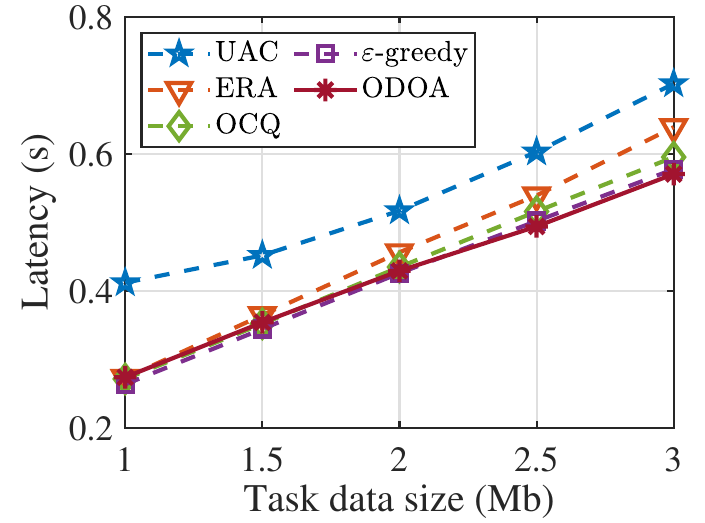}
		\end{minipage}
	}
	\subfigure[]
	{
		\begin{minipage}[t]{0.23\linewidth}
			\centering
			\includegraphics[scale=0.38]{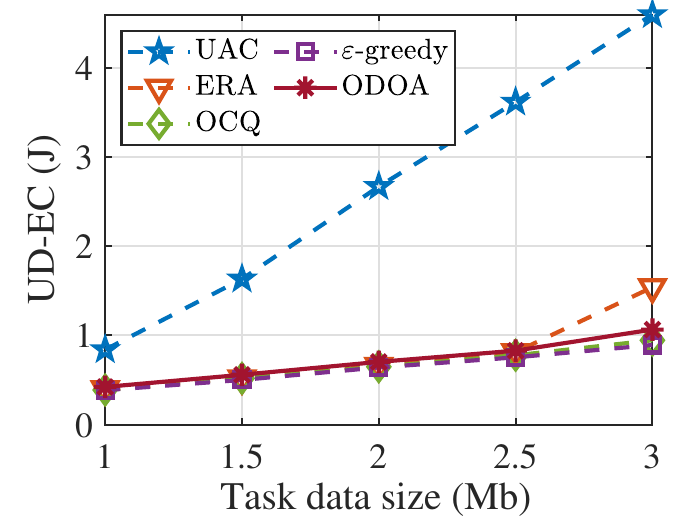}
		\end{minipage}
	}
        \subfigure[]
	{
		\begin{minipage}[t]{0.23\linewidth}
			\centering
			\includegraphics[scale=0.38]{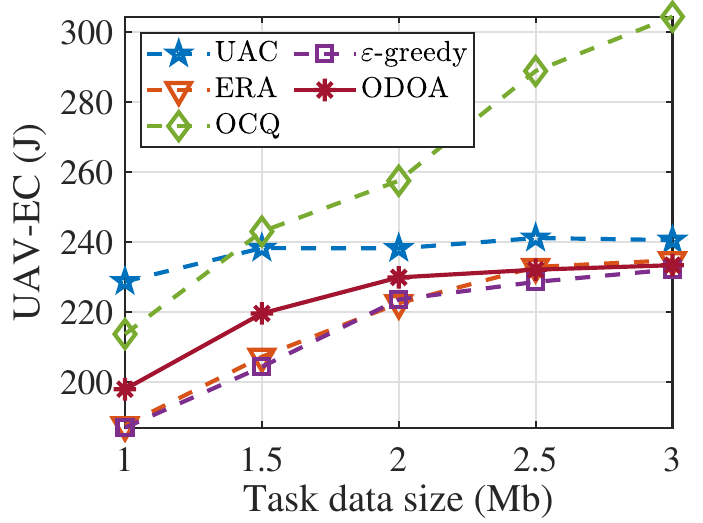}
		\end{minipage}
	}
	\caption{The impact of task data size on system performance. (a) Time-average ISD cost (Cost). (b) Average task completion latency (Latency). (c) Time-average ISD energy consumption (ISD-EC). (d) Time-average UAV energy consumption (UAV-EC).} 
	\label{fig_tasksize}
	\vspace{-1em}
\end{figure*}
\subsubsection{Impact of Task Data Size} Figs. \ref{fig_tasksize}(a), \ref{fig_tasksize}(b), \ref{fig_tasksize}(c), and \ref{fig_tasksize}(d) show the impact of different task data sizes on various performance metrics. It can be seen that with the increase in task data sizes, there is an upward trend with respect to the time-averaged ISD cost, average task completion latency, time-averaged ISD energy consumption, and time-averaged UAV energy consumption. This is expected since the larger task data size leads to higher overheads on computing, communication, and energy consumption for both ISDs and the UAV. In addition, UAC exhibits a significant growth trend with respect to the time-averaged ISD cost and average task completion latency compared to the other approaches. This is because UAC relies heavily on the limited computing resources of the UAV, which becomes a bottleneck for system performance as the size of task data continues to increase. Finally, compared with UAC, ERA, OCQ and $\varepsilon$-greedy, the proposed ODOA exhibits superior performance with respect to time-averaged ISD cost as the task data size increases, and achieves performance improvements of approximately 18.9\%, 10.7\%, 4.1\%, and 1.2\% in terms of average task completion latency when the task data size reaches 3 Mb. 

\par In conclusion, the set of simulation results shows that the proposed ODOA can effectively adapt to heavily loaded scenarios, delivering overall superior performance.
\begin{figure*}[!hbt] 
	\centering
	\setlength{\abovecaptionskip}{1pt}%
	\setlength{\belowcaptionskip}{1pt}%
	\subfigure[]
	{
		\begin{minipage}[t]{0.23\linewidth}
			\centering
			\includegraphics[scale=0.38]{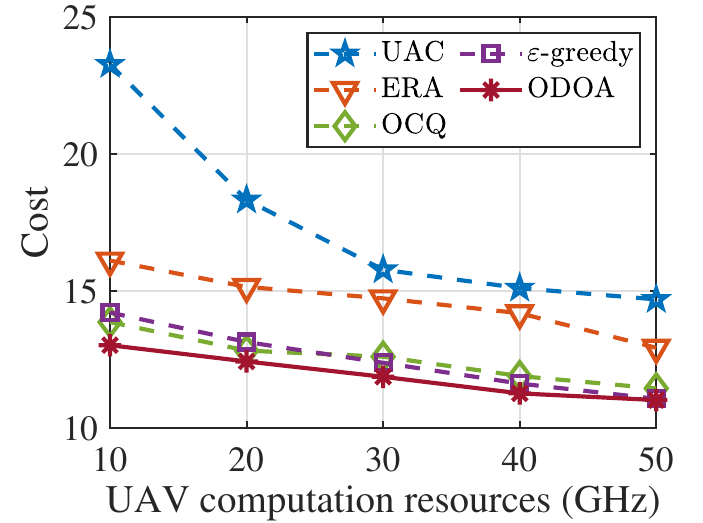}
		\end{minipage}
	}
	\subfigure[]
	{
		\begin{minipage}[t]{0.23\linewidth}
			\centering
			\includegraphics[scale=0.38]{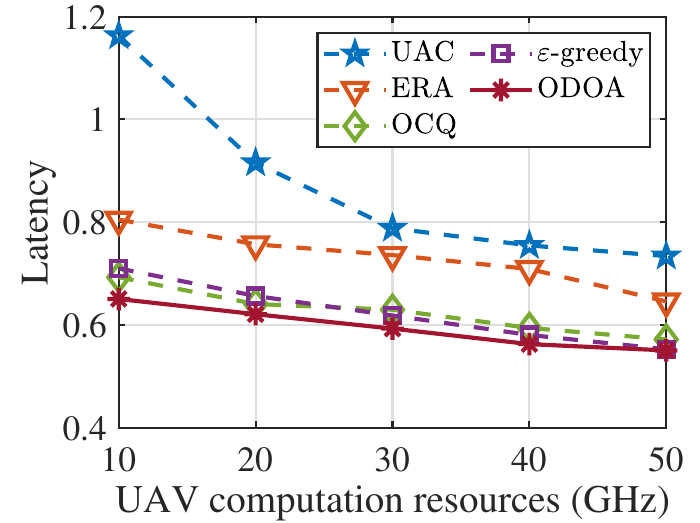}
		\end{minipage}
	}
	\subfigure[]
	{
		\begin{minipage}[t]{0.23\linewidth}
			\centering
			\includegraphics[scale=0.38]{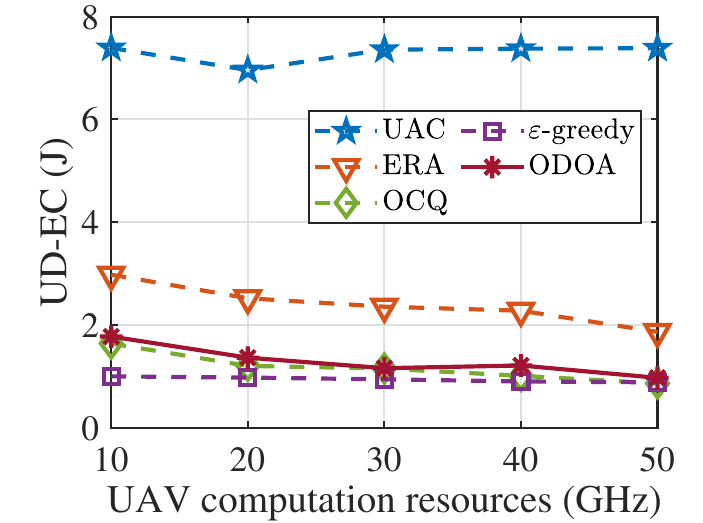}
		\end{minipage}
	}
        \subfigure[]
	{
		\begin{minipage}[t]{0.23\linewidth}
			\centering
			\includegraphics[scale=0.38]{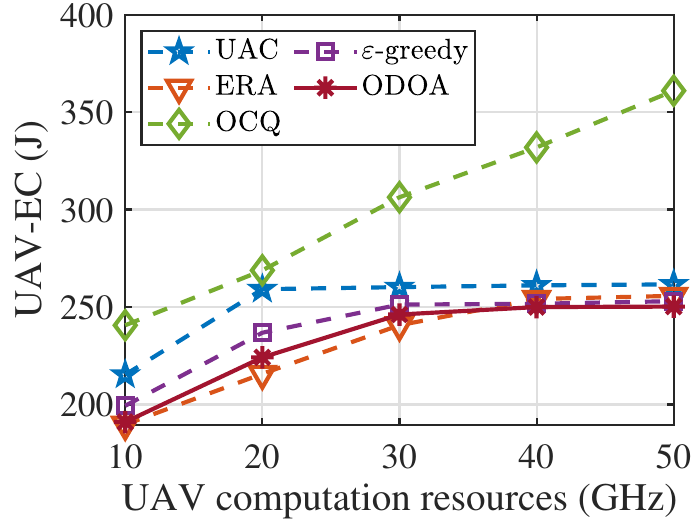}
		\end{minipage}
	}
	\caption{The impact of UAV computation resources on system performance. (a) Time-average ISD cost (Cost). (b) Average task completion latency (Latency). (c) Time-average ISD energy consumption (ISD-EC). (d) Time-average UAV energy consumption (UAV-EC).} 
	\label{fig_uav}
	\vspace{-1em}
\end{figure*}
\subsubsection{Impact of UAV Computation Resources} Figs. \ref{fig_uav}(a), \ref{fig_uav}(b), \ref{fig_uav}(c), and \ref{fig_uav}(d) compare the impact of different UAV computation resources on various performance metrics for the five approaches. It can be observed that with the increase of UAV computation resources, all approaches show a decreasing trend in terms of the time-averaged ISD cost and average task completion latency, and UAC and ODOA demonstrate gradually decreasing performance improvements. The reasons can be explained as follows. The increase of UAV computation resources provides more computing resource allocation for task execution, reducing the task execution latency. However, as UAV computation resources further increase, communication resources and the energy constraints of SUAV become bottlenecks that limit the improvement of system performance. Furthermore, EO maintains nearly constant performance in terms of the time-averaged ISD energy consumption regardless of the variations in UAV computing resources. This is mainly because the entire offloading of UAC primarily incurs transmission energy consumption for ISDs, which is independent of the UAV computation resources. 

\par Finally, ODOA outperforms UAC, ERA, OCQ, and $\varepsilon$-greedy in terms of the time-averaged ISD cost and average task completion latency, which illustrates the proposed approach enables sustainable utilization of computing resources and prevents resource over-utilization.

%
%
\section{Conclusion}
\label{sec:Conclusion}
\par In this paper, we explored the integration of computation offloading, satellite selection, resource allocation, and UAV trajectory control in digital twin-assisted SAGIMEC to advance the development of the LAE. We formally formulated the optimization problem to maximize the QoS for all ISDs. To solve this complex problem, we proposed ODOA, which combines the Lyapunov optimization framework for online control, the MAB model for online learning, and game theory for decentralized algorithm design. The mathematical analysis demonstrated that the ODOA not only meets the UAV energy constraint but also features low computational complexity. The simulation results indicate that the ODOA exhibits overall superior performance with respect to the time-averaged ISD cost, average task completion latency, and time-averaged ISD energy consumption.

%
%

\bibliographystyle{IEEEtran}
\bibliography{myref}
\begin{IEEEbiography}[{\includegraphics[width=1in,height=1.25in,clip,keepaspectratio]{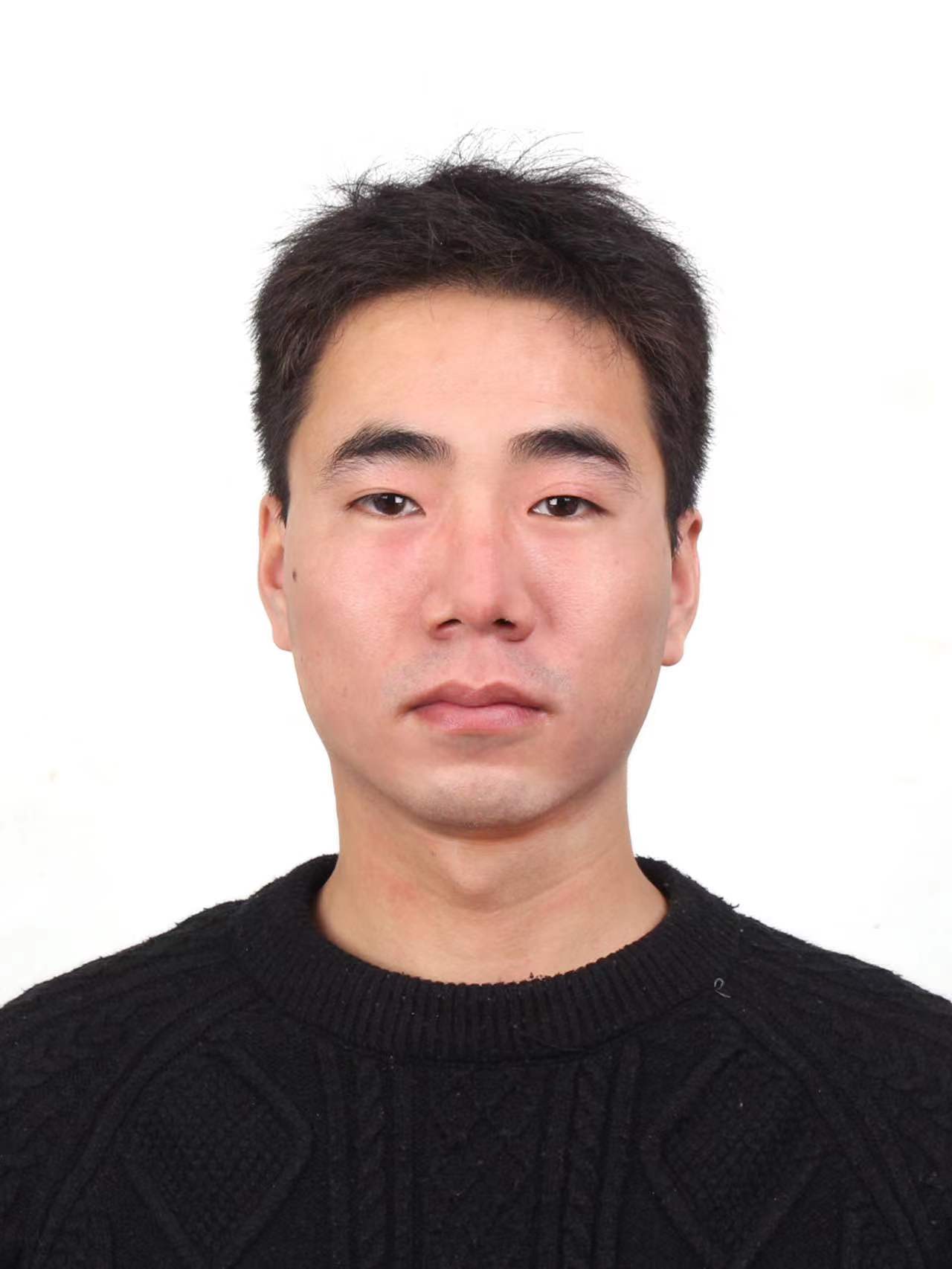}}]{Long He} received a BS degree in Computer Science and Technology from Chengdu University of Technology, Sichuan, China, in 2019. He is currently working toward the PhD degree in Computer Science and Technology at Jilin University, Changchun, China. His research interests include vehicular networks and edge computing.
\end{IEEEbiography}

\begin{IEEEbiography}[{\includegraphics[width=1in,height=1.25in,clip,keepaspectratio]{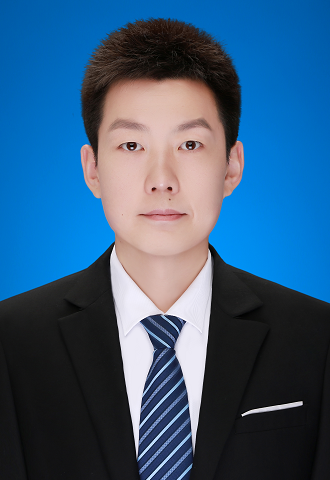}}]{Geng Sun} (S'17-M'19) received the B.S. degree in communication engineering from Dalian Polytechnic University, and the Ph.D. degree in computer science and technology from Jilin University, in 2011 and 2018, respectively. He was a Visiting Researcher with the School of Electrical and Computer Engineering, Georgia Institute of Technology, USA. He is a Professor in College of Computer Science and Technology at Jilin University, and his research interests include wireless networks, UAV communications, collaborative beamforming and optimizations.
\end{IEEEbiography}

\begin{IEEEbiography}[{\includegraphics[width=1in,height=1.25in,clip,keepaspectratio]{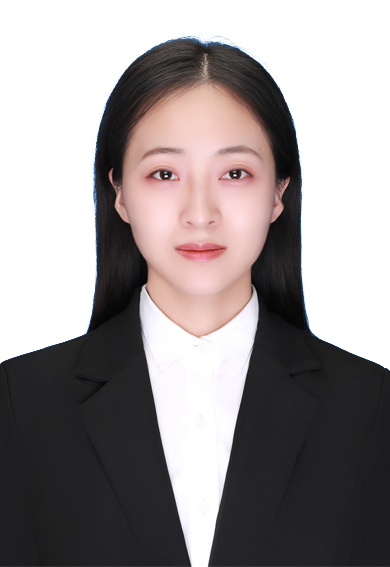}}]{Zemin Sun} received a BS degree in Software Engineering, an MS degree and a Ph.D degree in Computer Science and Technology from Jilin University, Changchun, China, in 2015, 2018, and 2022, respectively. Her research interests include vehicular networks, edge computing, and game theory. 
\end{IEEEbiography}

\begin{IEEEbiography}[{\includegraphics[width=1in,height=1.25in,clip,keepaspectratio]{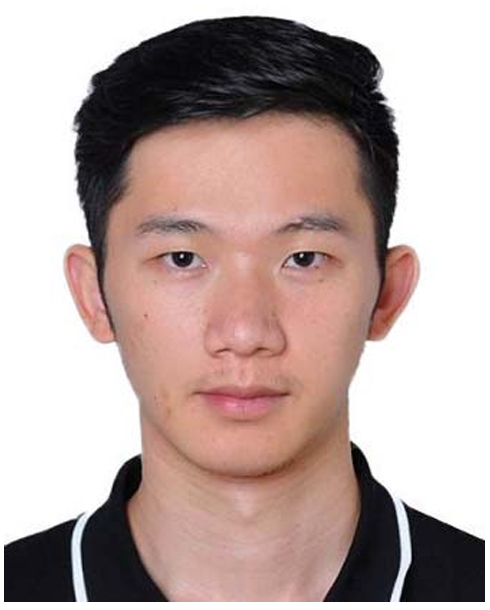}}]{Jiacheng Wang} received a bachelor’s degree from the Department of Science, Kunming University of Science and Technology in 2015, and the M.E. and Ph.D. degrees from the Department of Communication and Information Technology, Chongqing University of Posts and Telecommunications in 2018 and 2022, respectively. He is currently a Research Associate in computer science and engineering with Nanyang Technological University, Singapore. His research interests include wireless sensing, semantic communications, and metaverse.
\end{IEEEbiography}

\begin{IEEEbiography}[{\includegraphics[width=1in,height=1.25in,clip,keepaspectratio]{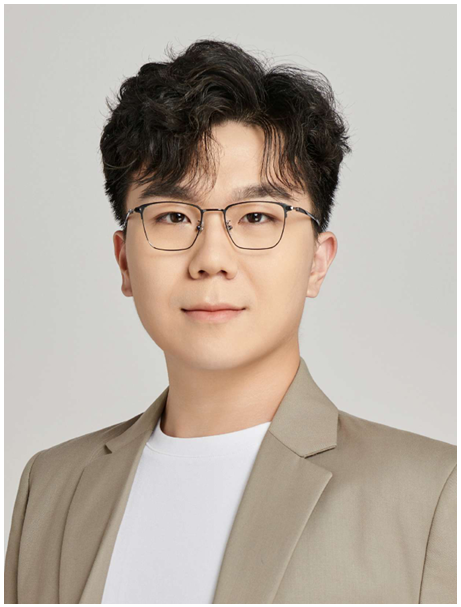}}]{Hongyang Du} is an assistant professor at the Department of Electrical and Electronic Engineering, The University of Hong Kong. He received the BEng degree from the School of Electronic and Information Engineering, Beijing Jiaotong University, Bei jing, and the PhD degree from the Interdisciplinary Graduate Program at the College of Computing and Data Science, Energy Research Institute @ NTU, Nanyang Technological University, Singapore. He serves as the Editor-in-Chief assistant of IEEE Communications Surveys \& Tutorials (2022-2024), the Editor of IEEE Transactions on Vehicular Technology, and the Guest Editor for IEEE Vehicular Technology Magazine. He is the recipient of the IEEE Daniel E. Noble Fellowship Award from the IEEE Vehicular Technology Society in 2022, the IEEE Signal Processing Society Scholarship from the IEEE Signal Processing Society in 2023, the Singapore Data Science Consortium (SDSC) Dissertation Research Fellowship in 2023, and NTU Graduate College’s Research Excellence Award in 2024. He was recognized as an exemplary reviewer of the IEEE Transactions on Communications and IEEE Communications Letters. His research interests include edge intelligence, generative AI, semantic communications, and network management.
\end{IEEEbiography}

\begin{IEEEbiography}[{\includegraphics[width=1in,height=1.25in,clip,keepaspectratio]{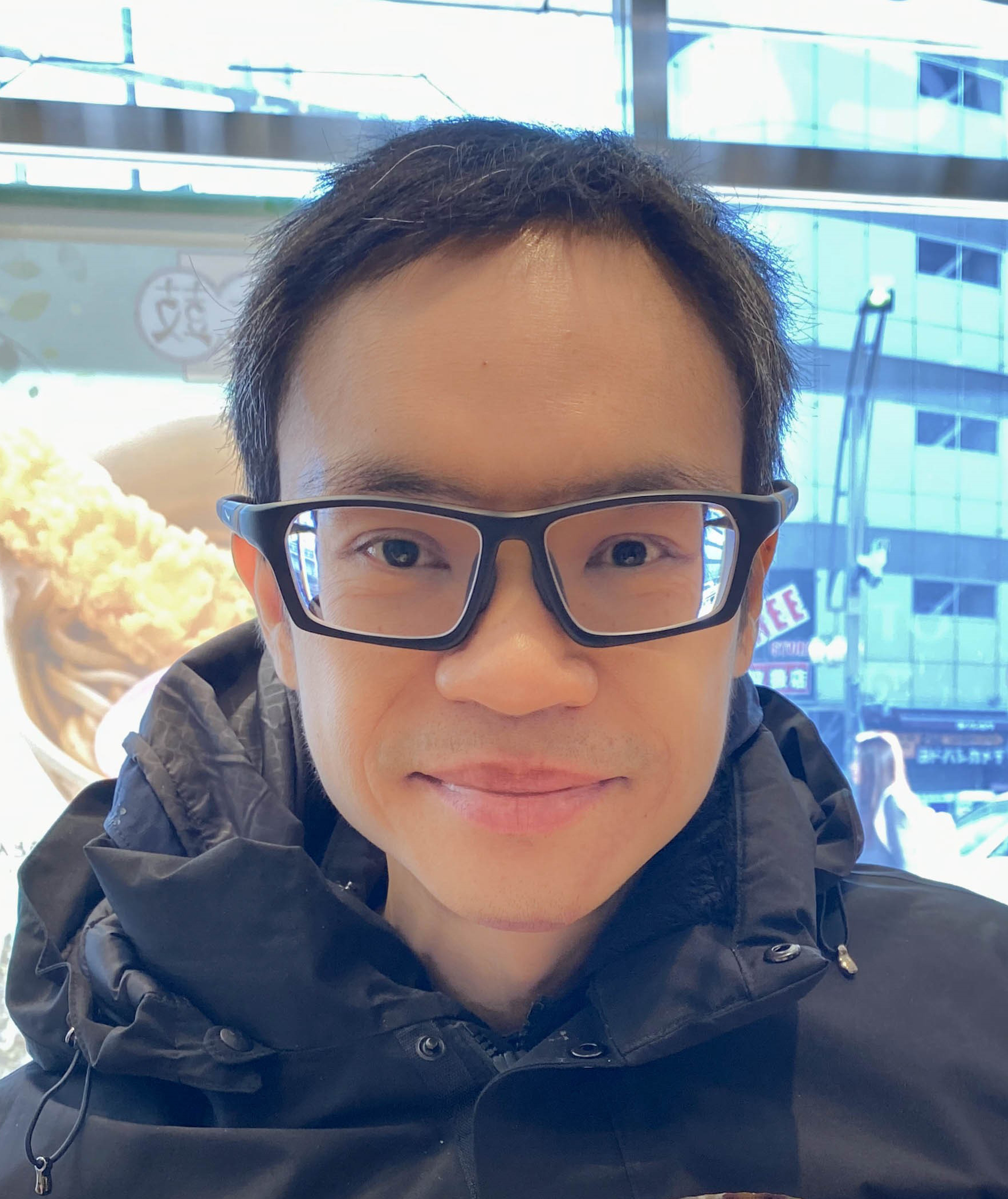}}]{Dusit Niyato} (Fellow, IEEE) received the B.Eng. degree from the King Mongkuts Institute of Technology Ladkrabang (KMITL), Thailand, in 1999, and the Ph.D. degree in electrical and computer engineering from the University of Manitoba, Canada, in 2008. He is currently a Professor with the School of Computer Science and Engineering, Nanyang Technological University, Singapore. His research interests include the Internet of Things (IoT), machine learning, and incentive mechanism design.
\end{IEEEbiography}

\begin{IEEEbiography}[{\includegraphics[width=1in,height=1.25in,clip,keepaspectratio]{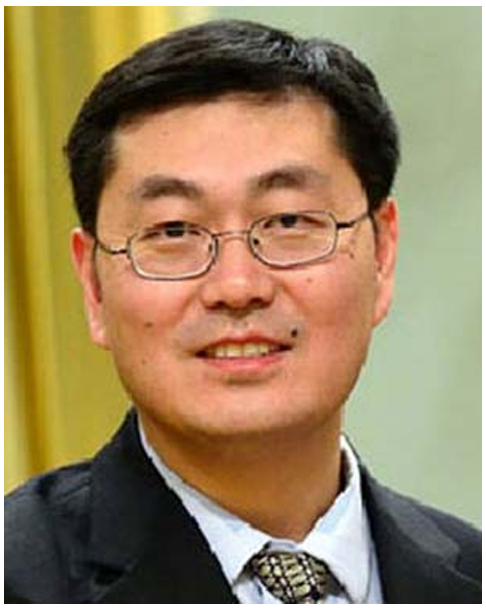}}]{Jiangchuan Liu} (Fellow, IEEE) received the BEng (cum laude) degree in computer science from Tsinghua University, Beijing, China, in 1999, and the PhD degree in computer science from the HongKong University of Science and Technology, in 2003. He is currently a full professor (with University Professorship) with the School of Computing Science, Simon Fraser University, BC, Canada. He is a fellow of the Canadian Academy of Engineering and the NSERC E.W.R. Steacie Memorial fellow. He is a steering committee member of IEEE Transactions on Mobile Computing. He was a co-recipient of the Test of Time Paper Award of the IEEE INFOCOM, in 2015, the ACM TOMCCAP Nicolas D. Georganas Best Paper Award, in 2013, and the ACM Multimedia Best Paper Award, in 2012. He is an associate editor of the IEEE/ACM Transactions on Networking, the IEEE Transactions on Big Data,and the IEEE Transactions on Multimedia.
\end{IEEEbiography}

\begin{IEEEbiography}[{\includegraphics[width=1in,height=1.25in,clip,keepaspectratio]{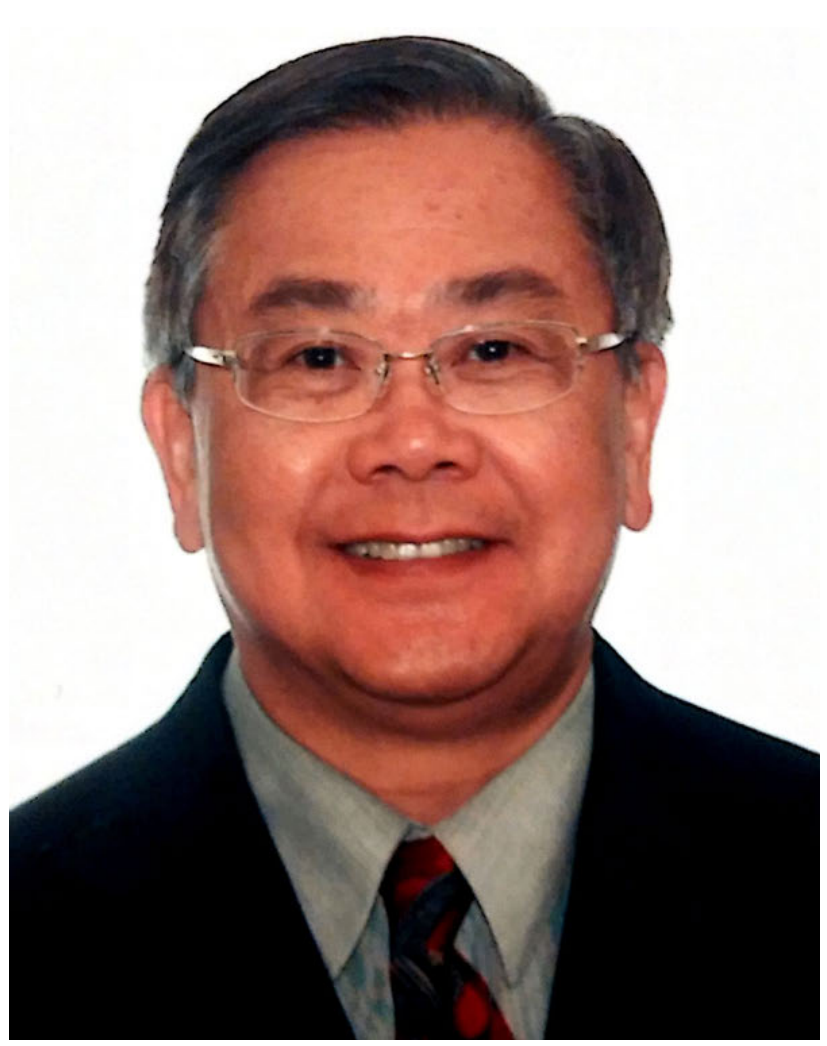}}]{Victor C. M. Leung} (Life Fellow, IEEE) is a Distinguished Professor of computer science and software engineering with Shenzhen University,
China. He is also an Emeritus Professor of electrial and computer engineering and the Director of the Laboratory for Wireless Networks and Mobile Systems at the University of British Columbia (UBC). His research is in the broad areas of wireless networks and mobile systems. He has co-authored more than 1300 journal/conference papers and book chapters. Dr. Leung is serving on the editorial boards of IEEE Transactions on Green Communications and Networking, IEEE Transactions on Cloud Computing, IEEE Access, and several other journals. He received the IEEE Vancouver Section Centennial Award, 2011 UBC Killam Research Prize, 2017 Canadian Award for Telecommunications Research, and 2018 IEEE TCGCC Distinguished Technical Achievement Recognition Award. He co-authored papers that won the 2017 IEEE ComSoc Fred W. Ellersick Prize, 2017 IEEE Systems Journal Best Paper Award, 2018 IEEE CSIM Best Journal Paper Award, and 2019 IEEE TCGCC Best Journal Paper Award. He is a Life Fellow of IEEE, and a Fellow of the Royal Society of Canada, Canadian Academy of Engineering, and Engineering Institute of Canada. He is named in the current Clarivate Analytics list of Highly Cited Researchers.
\end{IEEEbiography}
\end{document}